\documentclass[11pt, one column]{article}

\usepackage[T1]{fontenc}

\usepackage[margin=1in]{geometry}

\usepackage{graphicx}
\usepackage{hyperref}
\usepackage{color}

\usepackage{mathtools}
\usepackage{amssymb,amsmath,amsfonts,amstext,amsthm}
\usepackage{tikz}  
\usepackage[shortlabels]{enumitem}
\usepackage{authblk}
\usepackage{orcidlink}
\usepackage{cite}

\newtheorem{theorem}{Theorem}[section]
\newtheorem{corollary}{Corollary}[section]
\newtheorem{lemma}{Lemma}[section]
\newtheorem{proposition}{Proposition}[section]
\newtheorem{definition}{Definition}[section]
\newtheorem{example}{Example}[section]

\allowdisplaybreaks

\title{Hotelling-Downs with Facility Synergy:\\ The Mall Effect}

\author{Elliot Anshelevich\orcidlink{0000-0001-9757-6839}}
\author{Jianan Lin\orcidlink{0009-0007-3290-8864}}
\author{Noah Prisament\orcidlink{0009-0006-0989-0689}}
\affil{Rensselaer Polytechnic Institute (RPI), Troy, NY}
\affil{
    \href{mailto:eanshel@cs.rpi.edu}{\texttt{eanshel@cs.rpi.edu}},
    \href{mailto:hcm6755@gmail.com}{\texttt{hcm6755@gmail.com}},
    \href{mailto:nprisament@gmail.com}{\texttt{nprisament@gmail.com}}
}

\date{\today}

\begin{document}
\maketitle
\begin{abstract}
We consider a variation of the classic Hotelling-Downs model with the addition of facility synergies. Unlike in the classic model, where clients always use the facility closest to them, we study clients who prefer locations with many facilities to those with few facilities while simultaneously attempting to minimize their distance as well. We show that, in contrast with the classic model, Nash equilibria for our setting always exist, and, in fact, there always exists a Nash equilibrium such that the sum of client costs equals the cost of the optimal solution. Our main result is a bound of $\frac{225}{64}\approx 3.516$ on the Price of Anarchy for our model, showing that, although the client behavior is more complex in our model (and often more realistic depending on the application), the cost of Nash equilibrium solutions still cannot be much worse than the cost of the optimal facility placement.

\end{abstract}
\raggedbottom

\section{Introduction}
Beginning with Hotelling's seminal work almost a century ago \cite{Harold_Hotelling}, the Hotelling-Downs model has grown to become one of the classics of game theory. Originally defined using the story of ice cream vendors choosing their placement along a beach, it has been since used to model a variety of applications, including the strategic placement of facilities and firms, as well as the formation of political parties and the ideological placement of strategic political candidates.

In most existing work on the topic, there are {\em clients} located along some interval $X$, and a set of $k$ {\em facilities} or firms. Each facility can individually choose a location $x_i\in X$ where it will be placed; it is allowed to choose any location in $X$. In the classic model, each client simply utilizes the closest facility to them, i.e., a client located at $x$ would shop at the position $x_i$ with a facility such that $|x-x_i|$ is as small as possible. Note that in political and social choice applications, this corresponds to a voter/client supporting the party/candidate which is closest to them. Knowing this client behavior, the goal of each facility is typically to position itself to maximize the total fraction of clients using their facility. If multiple facilities choose the same location, then they typically share the clients equally. More precisely, let $u_i$ represent the fraction of clients which use location $x_i$, i.e., the fraction of clients for whom $x_i$ is the closest location of any point with a facility at it; and let $n_i$ be the number of facilities at $x_i$. Then, the {\em utility} of a facility located at $x_i$ is defined to be $u_i/n_i$: it is the total fraction of clients which use that location, divided by the number of facilities located there.

Many important variations of the above classic model have been studied for many diverse applications, as we discuss in the Related Work section. However, an important aspect missing from existing work is that of {\em facility synergy}. In many realistic applications, a client does not simply wish to go to the closest facility but prefers to go to a location with {\em multiple} facilities if possible. Consider, for example, someone who wishes to go shopping for a luxury good, such as designer shoes. They could go to a nearby store, but they may not find something that suits them. Or they could instead go to a collection of stores all near each other, even if these stores are farther away. Similarly, when people need to buy goods without knowing exactly what they are looking for, they are much more likely to go to a shopping mall with its many options than simply to the nearest store. Such behavior does not only apply to shopping, of course. When a group of friends wants to go to a restaurant or for entertainment but does not want to decide on a place in advance, they are likely to go to a place with a large collection of options, even if it is farther away. When joining a group (such as patients choosing a hospital or a medical practice, or graduate students choosing a research group), there is a benefit to joining a larger group since there will be many options and more support, even if there is a smaller group which is technically closer to your interests. This also applies to political parties: while a small party may exist which is closer to your views, people will often join larger parties due to the benefits which come from having a large number of powerful representatives who are members of your party.

In this paper, we begin the study of facility synergy by considering the simplest model possible, which we believe captures the basic fundamental differences with the classic model without synergy. To model the fact that clients prefer to go to a location with more facilities but still care about the distance as well, we define the client cost function $c(x,x_i)=|x-x_i|/n_i$ to be the cost that a client positioned at $x$ experiences for going to location $x_i$ with $n_i$ facilities located there. Each client then uses the location that minimizes their cost $c(x,x_i)$, not simply the distance. Thus, clients prefer to use closer facilities, but at the same time prefer locations which contain more facilities.
The utility of the facility $u_i$ is still defined as before: it is the total fraction of clients who use location $x_i$, and thus the utility that each facility receives is still just $u_i/n_i$.

\subsection{Our Contributions}
As done previously for the classic Hotelling-Downs model, our goal in this paper is to understand the structure of Nash equilibrium solutions which result from the self-interested behavior of the facilities. We are especially interested in the relationship of Nash equilibrium solutions with the optimal facility placement, as quantified by Price of Anarchy and Price of Stability measures.\footnote{The Price of Anarchy \cite{PoA_name} is the ratio between the cost of the worst Nash equilibrium and the cost of the optimal facility placement, while the Price of Stability \cite{anshelevich2008price} is the same ratio but for the best Nash equilibrium instead of the worst one.} In other words, we are interested in quantifying how much the self-interest of the facilities hurts the clients as compared to facilities being placed by some central and altruistic decision maker. Here, the {\em cost} of a solution is defined as the total cost of all the clients, as in most previous work.\footnote{Note that the total sum of the utilities of the facilities for our model is always the same since the game is constant sum, due to the fact that every client uses some facility location. This is in contrast to previous work such as \cite{Variations}. However, the cost of the clients can vary greatly depending on the solution.}

Although similar to the classic model, the difference in client behavior in our model results in very different properties of Nash equilibria and behavior of the facilities. This mainly stems from the fact that, unlike in the classic model where a set of clients using a particular facility location always forms a contiguous interval, in our model these sets can be discontinuous. In other words, while some interval of clients next to a set of facilities at $x_i$ will indeed use location $x_i$, there can be other intervals far away from $x_i$ which will also use $x_i$, as we discuss in Section \ref{sec:NE}. We call such intervals {\em bubbles}. The existence of bubbles greatly changes the structure of Nash equilibrium solutions and requires new techniques for their analysis. Another notable difference which complicates the analysis is that in the classic model, an addition of an extra facility always decreases the utilities of all the other facilities; in our model, however, it can decrease the utilities of some, but improve the utilities of others due to clients taking into account facility synergy. Thus, the facility utilities are non-monotone with respect to the addition of other facilities.

Despite the added complexity of the model, we show that in some ways the Nash equilibrium properties are improved in our setting. {\em We first prove that Nash equilibrium solutions always exist in our model, and in fact that the Price of Stability always equals 1, i.e., that there always exists a Nash equilibrium with the total client costs being as small as possible.} In contrast, it is well-known that a Nash equilibrium {\em does not} always exist in the classic Hotelling-Downs model \cite{Eaton_Lipsey}, and further that even when it does exist, it can be as much as a factor of 2 larger in cost than the optimal solution.

We then proceed to our main result, which is an analysis of Price of Anarchy for our model. In the classic Hotelling-Downs model, where clients simply go to the closest facility, the Price of Anarchy is at most 2 \cite{Networks}. Despite the existence of bubble intervals and other differences in our model behavior, we are able to {\em prove a bound of $\frac{225}{64}\approx 3.516$ on the Price of Anarchy in our model.} Doing this requires showing various structural results about the properties of Nash equilibrium solutions, such as that the size of bubble intervals is bounded compared to adjacent intervals. Our Price of Anarchy bound shows that although the client behavior is more complex in our model (and often more realistic depending on the application), the cost of Nash equilibrium solutions still cannot be too much worse as compared with the optimal facility placement.

\subsection{Related Work}
Since Hotelling's seminal work, \textit{Stability in Competition} \cite{Harold_Hotelling}, countless variations of the original model have been studied in depth. Downs extended the Hotelling model to voting patterns in a space of political ideology, realizing the term Hotelling-Downs model \cite{Downs}. The model was again expanded upon by \cite{Eaton_Lipsey} who extended the model from strictly analyzing duopolies towards arbitrary $n$-facility markets and characterized when a pure Nash equilibrium exists and when it does not.\footnote{This characterization was recently amended by \cite{bhaskar2024equilibrium}, who also gave numerous results about the computational complexity of finding a Nash equilibrium.} This established that the classic model lacks any pure Nash equilibrium for $n=3$, has a unique pure Nash equilibrium for $n=2,4,5$ and has an infinite number of pure Nash equilibria for $n>5$. Due to the popularity of this model, we refer to the recent survey by \cite{drezner2024competitive} (as well as \cite{Survey_Spatial_Competition,thisse_1992,Survey_Competitive_Location_Models}, and \cite{Survey_Location_Analysis}) for a more complete overview of the breadth of traditional Hotelling-Downs models that have been studied. These surveys denote trends in model variants, including Graitson's enumerated changes in the number of firms, the shapes of the demand curves, and the types of spaces \cite{Survey_Spatial_Competition}.  Depending on the literature, these variants have sometimes been referred to as \textit{Facility Location Games}, \textit{Competitive Facility Location Games}, and \textit{Voronoi Games} --- all with essentially the same meaning; however, to minimize confusion, we will refer to them as Hotelling-Downs games. However, throughout this extensive body of work, we are not aware of any analysis of our client behavior in which clients prefer locations with more facilities. 

Many variations of the model involve changing the metric space, $X$, that clients and facilities occupy. This is often done with the extension to a circle instead of an interval, thereby inducing periodic boundary conditions, as studied by \cite{Eaton_Lipsey} and \cite{Salop_Circles}. Further, a logical next step includes location games in 2-dimensional spaces, e.g. \cite{2-dimensional}, or in nonlinear or discrete markets, including graphs \cite{Cycle_Graphs}, networks \cite{Networks}, and finite sets of locations \cite{Finite_Locations}.

Other variations change how facilities attract clients or how said clients choose the facility they will utilize. For example, \cite{Variations} propose a variation on the original Hotelling-Downs model in which all facilities have a limited attraction interval and ``the support of clients that fall in the attraction interval of several agents is randomly shared among the latter.'' Similarly, \cite{Limited_Attraction} study a model akin to that of \cite{Variations} formalized with an attraction width $w_i$ for each facility and later extend this model to results over arbitrary distributions of clients. These models can both be considered variants of the general set of probabilistic \textit{Shapley Facility Location Games} from \cite{Shapley_Games}. A modification from the client perspective is captured by \cite{Random_Tolerance_Intervals}, who analyze a model where each client has a randomly distributed tolerance interval and they utilize the nearest facility in this interval, if one exists. This captures the concept that clients may not utilize the closest facility if it is too far from them. One type of variation that is close in spirit to our work is that of Hotelling-Downs models that take into account network externalities, such as the client cost function being a linear combination of distance and facility congestion as in \cite{Load_Balancing,peters2018hotelling}, or of distance and facility popularity as in \cite{fournier2024popularity}.

There are also similar models comprising of multi-unit and multi-stage games, including those in which single agents can place multiple facilities concurrently \cite{Multi_Unit} or in consecutive rounds \cite{The_Voronoi_Game}. Another multi-stage variation includes \cite{Waiting_Time} in which the first round involves facilities choosing locations on a graph and a second stage in which clients distribute their purchasing power. Likewise, multi-stage games include those where location choice is followed by additional differentiation such as price competition leading to ``non-symmetric'' facilities \cite{ECONOMIDES198667}. There is also a somewhat separate direction of research (see, e.g., \cite{chan2021mechanism} and the references therein) which uses a lot of similar terminology of facility location games but is concerned more with developing centralized mechanisms for placing all the facilities based on the locations reported by the clients, with the clients being able to lie about their true locations.

For these game-theoretic models, facility utility functions are generally either the more popular ``support maximizers'' function (utility is proportional to clients received) that is employed by our model or a ``winner-takes-all'' function that is generally reserved for modeling political landscapes \cite{Variations,Limited_Attraction}. Likewise, social welfare is generally modeled as minimizing the social cost due to client travel or prices, but it is occasionally defined as client ``participation'' in models that are not guaranteed to serve all clients as in \cite{Variations}. Lastly, the distribution of clients in the space is generally uniform across this body of work; however, there are exceptions that derive results for arbitrary or random distributions of clients such as \cite{Finite_Locations}, \cite{Price_of_Total_Anarchy}, and others.

Generally, the results derived in existing work focus on the existence and uniqueness of Nash equilibria, the optimal solutions, and the efficiency of said equilibria. A useful measure of the efficiency of equilibria are what are known as the Price of Stability for measuring the best-case efficiency \cite{anshelevich2008price}, and Price of Anarchy for measuring the worst-case efficiency \cite{PoA_first,PoA_name}. This is our focus as well.

\section{Our Model and Facility Synergy}\label{sec:model}
\paragraph{Classic Hotelling-Downs Model} First, let us recall the classic Hotelling-Downs model. We are given an interval $X$; without loss of generality, let us assume that $X=[0,1]$. There are $k$ facilities. Each facility can choose where in the interval $X$ it should be placed; it is allowed to choose any location in $X$. In the version of the problem we study, multiple facilities are allowed to place themselves at the same location; in this case, we say that a {\em stack} of $n_i$ facilities is located at $x_i\in X$. A {\em facility placement} is a set of locations $x_i$ and number of facilities $n_i$ at each of these locations, so that $\sum_i n_i= k$.\footnote{$i$ is formally defined as an arbitrary index from among all facility stacks.}

There are {\em clients} located in $X$; as in much of the previous work, we assume that the clients form a continuum and are uniformly distributed in $X$. In the classic model, the clients would utilize the closest facility to them, i.e., a client located at $x$ would use position $x_i$ with a facility so that $d(x,x_i)$ is as small as possible, where $d(x,x_i)$ is simply the distance $|x-x_i|$. For a given facility placement, let $U_i$ be the set of clients who use the location $x_i$, i.e., $U_i = \{x\in X : i = \arg\min_j d(x,x_j)\}$. Then $u_i$ (the size of $U_i$) is the total fraction of clients using location $x_i$.\footnote{Note that $U_i$ is usually an infinite set. By the ``size of $U_i$'', we mean the total fraction of clients in $U_i$ compared to $X$. For example, if $U_i$ consists of all clients in the interval $[\frac{2}{3},1]$, then $u_i$ would equal $\frac{1}{3}$. Formally, $u_i$ is defined as $u_i=\int_{x\in U_i}dx$.} The {\em utility} of a facility located at $x_i$ is defined to be $u_i/n_i$: it is the total fraction of clients which use that location, divided by the number of facilities located there, as clients going to a location are assumed to be equally shared between all the facilities at that location.

The clients are non-strategic, and simply use the facility closest to them. The facilities, on the other hand, choose their locations $x_i$ in order to maximize their utility as defined above. The choices made by the facilities on where they are positioned determines which locations the clients will use, and thus the utilities of the facilities themselves. In other words, the facilities are players in a game where they can choose any $x\in X$ as their strategy, and their utility is as defined above. A pure {\em Nash equilibrium} is a solution (i.e., facility placement) in which no single facility can increase their utility by changing their location.\footnote{In this paper, we will only focus on pure Nash equilibria, in which facilities must pick a specific location instead of a randomized strategy.} It is well known that a Nash equilibrium may not exist for this classic model \cite{Eaton_Lipsey}, and there are results about equilibrium quality when it does exist as well.

\paragraph{Adding Client Preferences for Facility Synergy} As discussed in the introduction, we instead consider clients who care both about the distance to $x_i$ and the number of facilities at $x_i$. We choose the simplest model which captures the essence of such clients, and define the client cost function $c(x,x_i)=d(x,x_i)/n_i$ to be the cost that a client at $x$ experiences for going to location $x_i$.\footnote{All our results hold in exactly the same way if we instead define $c(x,x_i)=d(x,x_i)/(\gamma n_i)$ for some constant $\gamma$. Changing $\gamma$ is equivalent to changing the size of the space $X$ which does not change the nature of the solutions and instead only changes the values of the costs and utilities.} Clients now choose to go to $x_i$ which minimizes the above cost, which increases with the distance but decreases with $n_i$. With this cost function, clients are indifferent between going to a location with 1 facility which is distance $y$ away, and going to a location with 2 facilities which is $2y$ away. As a first step toward modeling facility synergy, we consider this natural, since the travel cost per facility visited remains the same in both. In this new model, $U_i$ is still defined as the set of clients who choose to use $x_i$, that is, $U_i = \{x\in X : i = \arg\min_j c(x,x_j)\}$. $u_i$ is still defined as before, and thus, the utility that each facility located at $x_i$ receives is still just $u_i/n_i$. The only difference is which locations the clients choose.

\paragraph{Solution Cost and Price of Anarchy} We will study the existence of pure Nash equilibria as well as their quality. Many different measures have been considered in the literature for the quality of solutions in Hotelling-Downs and similar models. For our model (as well as the classic model defined above), the total utility of the facilities always equals $|X|=1$, since the clients always go somewhere.\footnote{We will use the notation $|I|$ to refer to the length (size) of an interval $I$.} The total cost of the clients, however, can be greatly impacted by how the facilities locate themselves. Because of this, as in the classic model, we consider as our objective function the total cost of the clients in the solution. More precisely, for a fixed choice of locations by the facilities, we can define the cost of a client at $x$ to be $$c(x)=\min_i c(x,x_i).$$
Then the total client cost is simply
$$\int_{x\in X}c(x)~dx.$$
Thus, the optimal facility locations are the ones which minimize the above quantity. It is not difficult to see that in the classic model where $c(x,x_i)=d(x,x_i)$, the optimal solution is simply to equally space the facilities inside the interval, although that solution is not a Nash equilibrium. However, for our model, where $c(x,x_i) = d(x,x_i)/n_i$, there can be many optimal solutions, as discussed in the next section.

We study both the Price of Anarchy and the Price of Stability of this game. The Price of Anarchy is the ratio between the cost of the {\em worst} (largest cost) Nash equilibrium and the cost of the optimal solution. It represents the possible harm experienced by the clients due to the self-interest of the facilities: if the facilities form a Nash equilibrium, this is how bad it can be as compared to their optimal placement, as it would be created by an altruistic central authority. We also consider the Price of Stability, which looks at the same ratio but uses the {\em best} Nash equilibrium. This represents the cost increase experienced by the clients if the resulting solution is required to be an equilibrium but could be chosen in order to minimize client cost. More formally, let $I$ be any instance of our game. Let $P(I)$ be the set of all possible solutions (facility placements) for instance $I$, let $N(I)\subseteq P(I)$ be the set of all pure Nash equilibria for an instance of the game $I$, and let $cost(S)$ be the total client cost in a solution $S$ for instance $I$ (i.e., the cost for a specific facility placement $S$).
Then the Price of Anarchy and Price of Stability of $I$ are defined as follows:

$$\text{Price of Anarchy} = \sup_{S\in N(I)}\frac{cost(S)}{\inf_{S'\in P(I)}cost(S')},$$
$$\text{Price of Stability} = \inf_{S\in N(I)}\frac{cost(S)}{\inf_{S'\in P(I)}cost(S')}.$$
The overall Price of Anarchy (Price of Stability) of a game is simply the worst-case Price of Anarchy (Price of Stability) over all instances $I$.

\section{Properties of the optimal solution}
In this section, we study the structure of the optimal facility placement which minimizes the total client cost and achieves the social optimum. First, we establish a lower bound on the cost of any solution.

\begin{lemma}\label{lem:opt_cost}
The cost of any solution is at least $\frac{1}{4k}$.
\end{lemma}
\begin{proof}
Consider an arbitrary solution, with stack $i$ located at $x_i$, containing $n_i$ facilities, and having utility $u_i$ (so each facility in stack $i$ has utility $u_i/n_i$). Then, the total cost of this solution is
$$\int_{x\in X} c(x)~dx = \sum_i \int_{x\in U_i} c(x,x_i)~dx =  \sum_i \int_{x\in U_i} \frac{|x-x_i|}{n_i}~dx.$$ 
The total client cost in $U_i$ is equal to $\int_{x\in U_i} \frac{|x-x_i|}{n_i}~dx.$ If we are allowed to choose any set $U_i$ with the total size $u_i = \int_{x\in U_i}dx$ fixed, then the set minimizing this cost is exactly the contiguous interval $[x_i-\frac{u_i}{2},x_i+\frac{u_i}{2}].$ And the total cost of clients in such an interval is exactly
$$\int_{x\in U_i} \frac{|x-x_i|}{n_i}~dx = 2\int_{x=0}^{u_i/2} \frac{x}{n_i}~dx = \frac{1}{n_i}(\frac{u_i}{2})^2 = \frac{u_i^2}{4n_i}.$$
Thus, the total cost of clients in $U_i$ is always at most $\frac{u_i^2}{4n_i}$, and the cost of any solution is at least 
$$\sum_i \int_{x\in U_i} c(x,x_i)~dx =  \sum_i \int_{x\in U_i} \frac{|x-x_i|}{n_i}~dx\geq \sum_{i}\frac{u_i^2}{4n_i}.$$
Due to the definition of $u_i$, we know that $\sum_i u_i = 1$ (or more generally, equal to the length of $X$). We also know that $\sum_i n_i = k$ since there are a total of $k$ facilities. For fixed values of $n_i$, setting each $u_i=\frac{n_i}{k}$ minimizes the above function, indicating that the solution cost is lower bounded by
$$\sum_{i}\frac{(\frac{n_i}{k})^2}{4n_i}=\sum_{i}\frac{n_i}{4k^2}=\frac{k}{4k^2}=\frac{1}{4k}.$$
\end{proof}

It is not difficult to see that any proportional spacing of facilities, as defined below, results in an optimal solution. 

\begin{definition}
A {\bf proportional spacing} is defined as follows. Given a list of stacks of facilities where the number of facilities in each stack is $n_i$, we can iterate through the list and place each stack at position: $$x_i=\frac{n_i}{2k}+ \sum_{j=1}^{i-1}\frac{n_j}{k}.$$
\end{definition}

\begin{proposition}
    \label{claim:proportional} Any proportional spacing has optimal cost. 
\end{proposition}
\begin{proof}
It is not difficult to verify that for any proportional spacing, the set of clients $U_i$ using stack $i$ is always the contiguous interval
$$U_i=\left[\sum_{j=1}^{i-1}\frac{n_j}{k},\frac{n_i}{k}+\sum_{j=1}^{i-1}\frac{n_j}{k}\right].$$
Thus, $u_i=n_i/k$, and the utility of each facility is exactly $u_i/n_i = 1/k$. Moreover, the client cost of any proportional spacing is exactly $\frac{1}{4k}$. To see why this is true, note that in a proportional spacing, stack $i$ is located exactly at the midpoint $x_i$ of interval $U_i$. Thus, the total client cost for interval $U_i$ is $$\int_{x\in U_i} \frac{|x-x_i|}{n_i}~dx = 2\int_{x=0}^{n_i/2k} \frac{x}{n_i}~dx = \frac{1}{n_i}(\frac{n_i}{2k})^2 = \frac{n_i}{4k^2}.$$ Summing this up over all stacks $i$, we have a total client cost of $$\frac{\sum_i n_i}{4k^2} = \frac{k}{4k^2} = \frac{1}{4k}.$$
Combined with Lemma \ref{lem:opt_cost}, this tells us that every proportional spacing is always optimal.
\end{proof}

Proportional spacings are easy to analyze: they have no ``bubbles'' (i.e., the set of clients using a stack is always a contiguous interval), each stack obtains a total number of clients equal to $u_i=n_i/k$, and thus each facility has utility exactly $1/k$. In particular, the following natural placements of facilities are optimal.

\begin{corollary}\label{cor:central}
Placing all facilities in a stack of size $k$ at position $\frac{1}{2}$ results in an optimal solution.     
\end{corollary}

\begin{corollary}\label{cor:spaced}
Placing each facility in a separate stack, with the first at position $\frac{1}{2k}$ and the $i$'th at position $\frac{i-1}{k}+\frac{1}{2k}$ results in an optimal solution.
\end{corollary}

The latter solution is exactly the optimum for the classic Hotelling-Downs model as well. The former solution is the simplest to analyze: it is simply all the facilities teaming up to form a giant mall in the middle of the interval, so the cost of a client located at $x$ is exactly $|x-\frac{1}{2}|/k$. 

Although the optimal solutions in this model are not difficult to analyze, they are {\em not necessarily} Nash equilibrium solutions. Nevertheless, we can establish the following claim, which shows the existence of Nash equilibria which are in contrast to the classic model.

\begin{proposition}\label{claim:PoS}
The solution from Corollary \ref{cor:central} is a Nash equilibrium. Thus, Nash equilibria always exist, and the Price of Stability is 1. 
\end{proposition}
\begin{proof}
The proof of Proposition \ref{claim:PoS} is mostly straight-forward. Consider the solution where all facilities are located at position $\frac{1}{2}$, and consider some facility which decides to move to position $y$ (assume that $y<\frac{1}{2}$ without loss of generality). Before moving, the facility has utility $1/k$. For $k\leq 2$ that facility clearly cannot obtain more utility than that by moving away from the center, so we can assume $k\geq 3$.

A client located at $x$ has a choice between using the facilities at $\frac{1}{2}$ at a cost of $\frac{|x-\frac{1}{2}|}{k-1},$ or using the facility at $y$ at a cost of $|x-y|$. Thus, the only clients which will use $y$ are the ones for which $|x-y| \leq \frac{|x-\frac{1}{2}|}{k-1}.$ This is not true for any $x\geq \frac{1}{2}$. For $x<\frac{1}{2}$, consider the values of $x$ for which the costs of using either stack are the same, i.e., $\frac{\frac{1}{2}-x}{k-1}=|x-y|$. This occurs at $x=\frac{\frac{1}{2} + (k-1)y}{k}$, and at $x=\frac{(k-1)y-\frac{1}{2}}{k-2}$. The second value of $x$ may be negative, however, and thus to the left of the starting point of the interval $X$. Thus, the interval of clients for whom it is cheaper to use the facilty at $y$ is exactly
$$\left[\max\left(\frac{(k-1)y-\frac{1}{2}}{k-2},0\right), \frac{\frac{1}{2} + (k-1)y}{k}\right].$$
Thus, the utility obtained by the facility which moved is the length of this interval. First, let us consider the case where $\frac{(k-1)y-\frac{1}{2}}{k-2} \geq 0$. Then, the utility obtained by the facility which moved to $y$ is exactly
$$\frac{\frac{1}{2} + (k-1)y}{k} - \frac{(k-1)y-\frac{1}{2}}{k-2} \leq \frac{\frac{1}{2} + (k-1)y}{k} - \frac{(k-1)y-\frac{1}{2}}{k} = \frac{1}{k}.$$
This means that by moving, the facility obtained a utility of at most $\frac{1}{k}$, which is how much it already had before moving to its new position. The same occurs for the case where $\frac{(k-1)y-\frac{1}{2}}{k-2} < 0$. In that case, we have that $y<\frac{1}{2(k-1)}$, since we assumed that $k\geq 3.$ Thus, the length of the above interval of clients, and thus the utility obtained by the facility that moved, is at most
$$\frac{\frac{1}{2} + (k-1)y}{k} \leq \frac{\frac{1}{2} + (k-1)\frac{1}{2(k-1)}}{k} = \frac{1}{k}.$$

Therefore, the facility which moved to $y$ cannot increase its utility to greater than $\frac{1}{k}$ by performing this move. Since the move was arbitrary, this means that the solution in which everyone is located at position $\frac{1}{2}$ is a Nash equilibrium.
\end{proof}

In fact, most of the proportional facility spacings are Nash equilibrium solutions. For example, if each stack is of size at least 2 (no facility is located by itself), then an optimal solution is a Nash equilibrium. On the other hand, the solution in Corollary \ref{cor:spaced} is not a Nash equilibrium, as the first and last facilities have incentive to move closer to the middle and thus obtain a larger share of the clients.\footnote{Similarly to all Nash equilibria in the classic model, the first and last facilities can never be stacks of size 1 for an optimal solution to be a Nash equilibrium when $k>1$. This is a necessary but not sufficient condition for optimal equilibria.}

\section{Properties of Nash equilibria}\label{sec:NE}
Although the best Nash equilibrium solutions are easy to analyze, since they are the same as optimal solutions, looking at the properties of general Nash equilibria becomes complex. The main difference compared to the classic model is the existence of {\em bubbles}. To define these formally, we first need to introduce the notion of core intervals.

Recall that the cost function $c(x)$ is defined as $ \min_i\{c(x,x_i)\}$. Thus, the clients $U_i$ providing utility to stack $i$ are exactly the ones at locations $x$ where $c(x)=c(x,x_i)$. When there are several stacks $i$ with costs to them being tied, the clients can provide their utility to any of them; since the clients form a continuous interval, this will only occur at a finite number of points and will not affect the total utility of any stack.

\begin{definition}
For each stack located at $x_i$, define its {\bf core interval} $I_i$ as the maximal contiguous interval such that $I_i\subseteq X$, $x_i\in I_i$ and all clients $x\in I_i$ are using facility stack $i$.
\end{definition}

In other words, core intervals are the sets of clients located next to a facility stack that use this stack. In the classic model, all clients belong to a core interval. In our model, however, client behavior results in bubbles, which are intervals of clients that do not belong to any core interval. To illustrate this, and show that this does, in fact, occur in Nash equilibrium solutions, consider the following simple example.

\begin{figure*}[ht]
        \centering
        \includegraphics[width=1\textwidth]{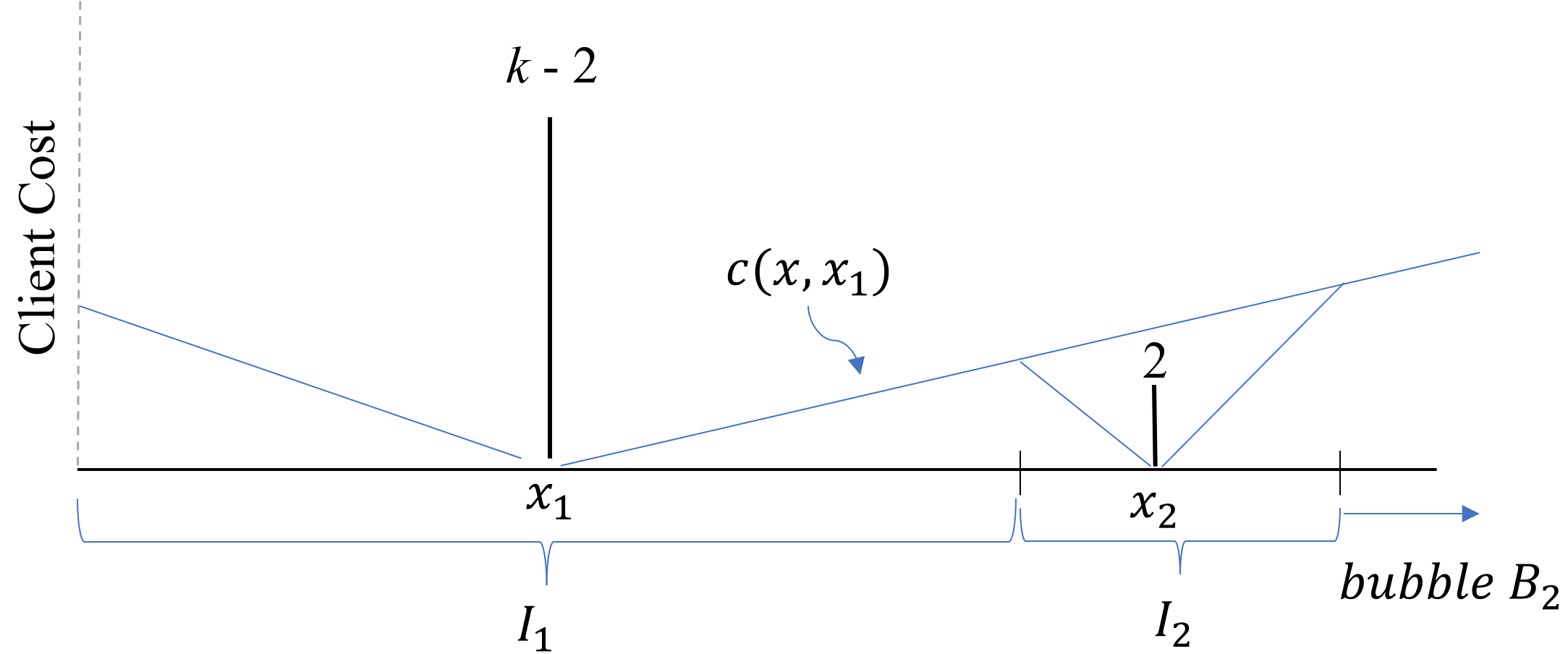}
        \caption{An 
        example with a bubble (Example \ref{ex:bubbles}). $k-2$ facilities are located at position $x_1$, and 2 facilities at location $x_2$. The blue lines show the client costs $c(x,x_1)=|x-x_1|/(k-2)$ to use facilities at $x_1$, and $c(x,x_2)=|x-x_2|/2$ to use facilities at $x_2$. The clients in interval $I_2$ use facilities at $x_2$, since that gives them the smallest cost. The clients both in intervals $I_1$ {\em and} in $B_2$ use facilities at $x_1$, since that has the smallest cost for them. Thus, the set $B_2$ forms a ``bubble'': a set of clients who use a facility stack such that they need to pass through another stack to get to it. In this figure, $u_1= |I_1|+|B_2|$ and $u_2 = |I_2|$.}
        \label{fig:example}
\end{figure*}

\begin{example} \label{ex:bubbles} {\bf[Bubble Example]} Consider the following solution, as shown in Figure \ref{fig:example}. For ease of presentation, let $|X|=2k-7-\varepsilon$; everything can be scaled equivalently to $|X|=1$. Position a stack of $n_1=k-2$ facilities at location $x_1=k-6$, and a stack of $n_2=2$ facilities at location $x_2=2k-10$. Now consider what stack the clients located to the right of $x_2$ end up using. 
It is not difficult to verify that the clients in the interval $[x_2,x_2+2]$ use the stack at $x_2$, but the clients in the interval $[x_2+2,x_2+3-\varepsilon]$ use the stack at $x_1$ again.\footnote{It does not matter what facility clients at the boundary between intervals (i.e., a location $y$ such that $c(y)=c(y,x_i)=c(y,x_j)$ where $i\neq j$) use, as they only contribute an infinitesimal amount of utility.} Thus, the clients in the latter interval form a bubble: they use the stack which is farther away from them than $x_2$, and are not part of $x_1$'s core interval. The clients using the stack at $x_1$ consist of many clients next to it (its core interval), as well as the clients on the other side of $x_2$'s core interval. 
\end{example}

The above example shows that there can be sets of clients which do not belong to any core interval, unlike in the more classic Hotelling model. To better understand what these intervals look like, we show the following simple lemma. While it certainly holds for the cost function $c(x)= \min_i\{c(x,x_i)\}$, it is convenient to prove it for more general functions, as defined below. 

\begin{definition}\label{def:c-A}
For an arbitrary set $A$, let $c_{-A}(x)= \min_{i\not\in A}\{c(x,x_i)\}$ be the cost function at position $x$ if all stacks in set $A$ were removed from the solution.
\end{definition}

\begin{lemma}
\label{monotonic}
Let $R$ be an interval such that the only stacks  located in $R$ are members of some set $A$. Then, the cost function $c_{-A}(x)= \min_{i\not\in A}\{c(x,x_i)\}$ on interval $R$ is either:
\begin{enumerate}
    \item monotonically increasing
    \item monotonically decreasing
    \item monotonically increasing and then decreasing
\end{enumerate}
\end{lemma}
\begin{proof}
Let $R=[a,b]$; by assumption we know that only facilities in $A$ may be located in $R$.
The lemma follows from the definition of the cost function $c_{-A}(x)$. Recall that $c(x,x_i)$ is a strictly decreasing linear function of $x$ for $x\leq x_i$, and a strictly increasing linear function of $x$ for $x\geq x_i$, with $c_{-A}(x)$ being the piecewise linear function which is the minimum of all $c(x,x_i)$ for $i\not\in A$. 

Now, suppose to the contrary that the cost function on interval $R$ does not follow one of the patterns stated in the lemma. This means that there must exist points $y_1<y_2$ in $R$ such that $c_{-A}(x)$ is decreasing at $y_1$ and increasing at $y_2$. Let $x_1$ be the stack such that $c_{-A}(y_1)=c(y_1,x_1)$, thus, this stack is not in $A$. Since $c_{-A}(x)$ is decreasing at $y_1$ we know that $y_1<x_1$. Similarly, let $x_2$ be the stack such that $c_{-A}(y_2)=c(y_2,x_2)$; we know that $x_2<y_2$. But since $R$ does not contain any facilities which are not in $A$, it must be that $x_2<a<y_1<y_2<b<x_1$. Thus, we know that $c(y_1,x_2)<c(y_2,x_2)$ (since $y_2$ is farther from $x_2$), and $c(y_1,x_1)>c(y_2,x_1)$. Since at $x=y_1$ we have $c(x,x_1)$ as the minimum of the client costs to any stack not in $A$, it must be that $c(y_1,x_1)\leq c(y_1,x_2)$. But this is a contradiction, since it implies that $c(y_2,x_1)<c(y_2,x_2)$, which contradicts $c_{-A}(y_2)$ being equal to $c(y_2,x_2)$; the minimum of the cost functions would be at most $c(y_2,x_1)$ instead. This completes the proof.
\end{proof}


\begin{corollary}
    Lemma \ref{monotonic} holds for the cost function $c(x)= \min_i\{c(x,x_i)\}$ between two neighboring core intervals.
\end{corollary}

\begin{proof}
Denote a region between two consecutive core intervals as $R=[a,b]$; by definition of core interval we know that no facilities are located in $R$. Then, setting $A=\emptyset$ and applying Lemma \ref{monotonic} to $c(x)=c_{-\emptyset}(x)$ and $R$ gives us the desired result.
\end{proof}

Using the above corollary, we can now define bubbles more precisely.

\begin{definition}
A {\bf bubble} is a maximal contiguous interval that is disjoint from all core intervals such that the bubble has either a monotonically increasing or a monotonically decreasing cost function. 
\end{definition}

Note that the piecewise linear cost function in a bubble can correspond to multiple different stacks, i.e., the clients in a bubble may be going to different stacks $x_i$ to minimize their costs, not necessaily all to the same one.

\begin{definition}
An {\bf up-bubble} is a bubble where the cost function is monotonically increasing.
\end{definition}

\begin{definition}
A {\bf down-bubble} is a bubble where the cost function is monotonically decreasing.
\end{definition}

Because of the existence of bubbles, we cannot use standard techniques to analyze the Price of Anarchy for our model. Moreover, even if there were no bubbles, the equilibrium solutions have different properties than those of the classic model. It is not difficult to verify, for example, that the solution in Corollary \ref{cor:central} remains a Nash equilibrium even if the stack is not located at exactly $\frac{1}{2}$, but instead deviates from $\frac{1}{2}$ by a small distance $\delta$. Because of this, we must develop a new approach for showing our Price of Anarchy bounds, as we do in Section \ref{sec:poa}, which involves bounding the size of various bubbles and relating them to neighboring core intervals. 

\subsection*{Notation and Cost Behavior in Core Intervals}
Before we proceed to bounding the Price of Anarchy, however, we prove several useful, but somewhat technical, lemmas about the properties of core intervals and bubbles. We also introduce various notation which will be heavily used in the rest of this paper. 

\begin{definition} Given a closed interval $I=[a,b]$, we define $L(I)=a$ and $R(I)=b$.
\end{definition}

Recall that we use $|I|$ to refer to the length (size) of an interval $I$, i.e., $|I|=R(I)-L(I)$.

\begin{definition}
Let $\Lambda_i= c(L(I_i))$, for any core interval $I_i$, be the cost of the left-most client of a core interval. Similarly, let $M_i=c(R(I_i))$ for any core interval $I_i$.
\end{definition}

\begin{definition}
Let $c_{-i}(x)= \min_{j\neq i}\{c(x,x_j)\}$ be the cost function at position $x$ if stack $i$ were removed from the solution.
\end{definition}

\begin{lemma}[The $\Lambda-M$ Lemma]
\label{LM lemma}
If there is an up-bubble immediately to the right hand (positive) side of an interval $I_i$, then $\Lambda_i<M_i$. Moreover, for any $x\in I_i$, let $\delta = x- L(I_i)$. Then $c_{-i}(x) \geq \Lambda_i+\delta\left[\frac{M_i-\Lambda_i}{|I_i|}\right]$ (i.e., $c_{-i}(x)$ in the interval $I_i$ is at least the linear interpolant of $\Lambda_i$ and $M_i$). The function $c_{-i}(x)$ is concave in the interval $I_i$.
\end{lemma}

\begin{figure*}[ht]
        \centering
        \includegraphics[width=0.8\textwidth]{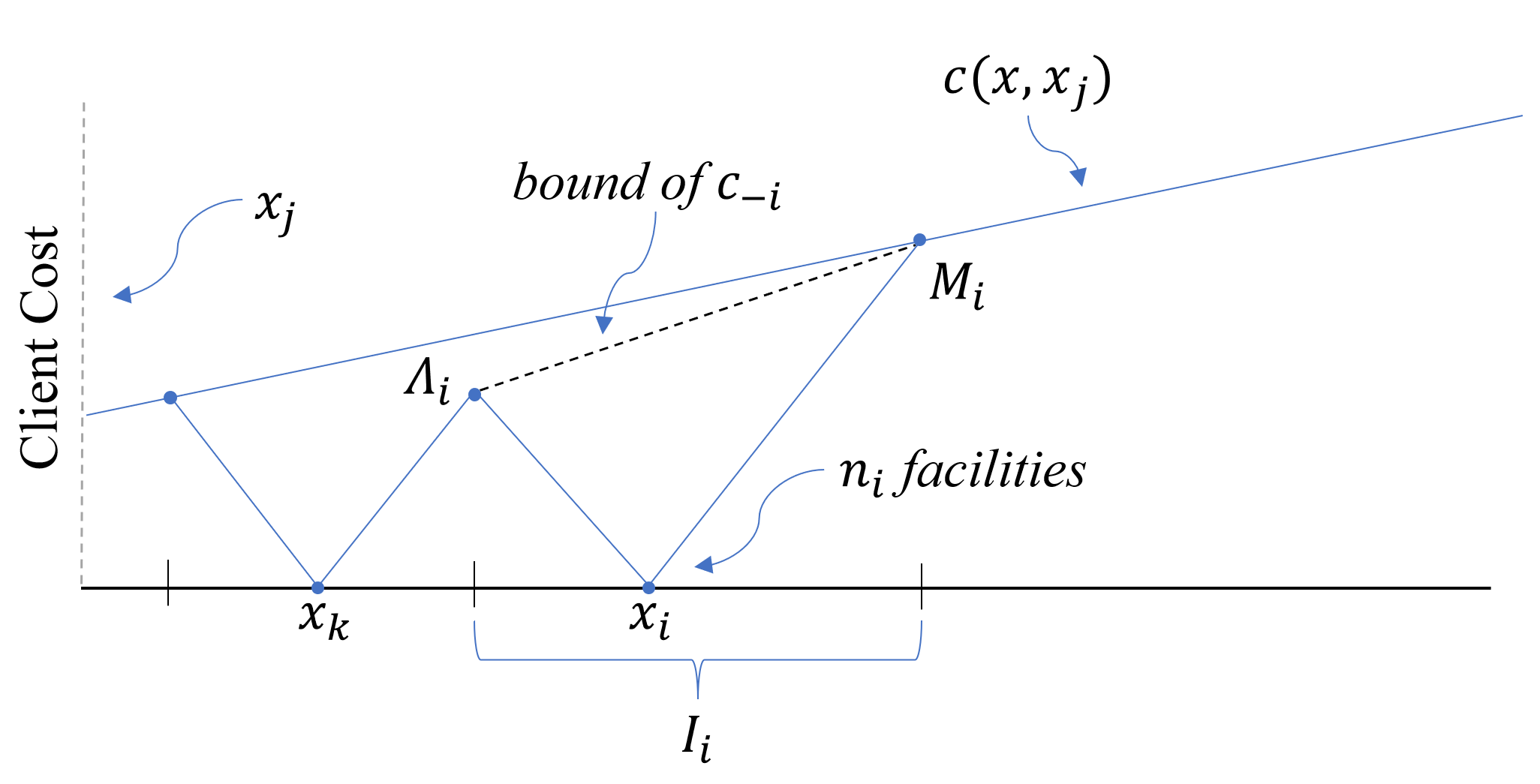}
        \caption{An illustration of the proof of Lemma \ref{LM lemma}.}
        \label{fig:LM}
\end{figure*}

\begin{proof}
See Figure \ref{fig:LM} for an illustration of this proof. Note that $c_{-i}$ is a piecewise linear function with a finite number of linear segments, each segment corresponding to the cost function of some stack. Consider the linear segment of $c_{-i}$ containing $R(I_i)$ and the points directly to the right of it (since the number of linear segments is finite, they must have non-zero length). In other words, let stack $j$ be the stack such that $c_{-i}(R(I_i))=c(R(I_i), x_j)$ and for some small $\epsilon > 0$ we have that $c_{-i}(R(I_i)+\epsilon)=c(R(I_i)+\epsilon, x_j)$. 

We know that stack $j$ must be to the left of $R(I_i)$ because stack $i$ is followed by an up-bubble and if $j$ were to the right of $R(I_i)$ the cost function would begin to decrease after $R(I_i)$ instead of increase. We also know that stack $j$ must be to the left of $L(I_i)$ or else $I_i$ would not be a core interval. We can show that $c(L(I_i), x_j) < c(R(I_i), x_j)$ because $c(x, x_j)$ is increasing from $L(I_i)$ to $R(I_i)$ because stack $j$ is to the left of $L(I_i)$. Thus, we know that $\Lambda_i$ must be less than $M_i$ because $c(L(I_i))\leq c(L(I_i), x_j)$ and $c(R(I_i))=c(R(I_i), x_j)$.

We know that $c_{-i}(L(I_i))=\Lambda_i$ and that $c_{-i}(R(I_i))=M_i$ by the definition of the cost function and $c_{-i}(x)$. Now apply Lemma \ref{monotonic} to the function $c_{-i}(x)$ and interval $I_i$. Since $I_i$ is a core interval, it cannot contain facilities other than the stack $x_i$. Thus, Lemma \ref{monotonic} tells us that $c_{-i}(x)$ is concave in this interval.
Therefore, using the concavity of $c_{-i}(x)$ alongside its values at the start and end of $I_i$, we can conclude that $c_{-i}(x) \geq \Lambda_i+\delta\left[\frac{M_i-\Lambda_i}{|I_i|}\right]$ by the definition of concave functions. This can be interpreted as $c_{-i}(x)$ is greater than the linear interpolant of $\Lambda_i$ and $M_i$ within $I_i$.
\end{proof}

While the above lemma is phrased for up-bubbles, the same holds for down-bubbles via a symmetric argument.

\begin{corollary}
\label{LM corollary}
If there is a down-bubble immediately to the left hand (negative) side of an interval $I_i$, then $M_i<\Lambda_i$. For any $x\in I_i$, let $\delta = R(I_i)-x$. Then $c_{-i}(x) \geq M_i-\delta\left[\frac{M_i-\Lambda_i}{|I_i|}\right]$ (i.e., $c_{-i}(x)$ in the interval $I_i$ is at least the linear interpolant of $\Lambda_i$ and $M_i$). The function $c_{-i}(x)$ is concave in the interval $I_i$.
\end{corollary}

The above lemma bounds what happens to the cost $c_{-i}$ within a core interval $I_i$ if it is adjacent to a bubble. Below, in Lemma \ref{lm lemma}, we ask the same question for the entire set $U_i$, which will contain the core interval $I_i$, but may also contain other clients in bubbles. 

\begin{definition}
Let $l_i$ be the leftmost/smallest $x$ value of all clients that choose to utilize facility stack $i$, i.e., the leftmost point of $U_i$. Similarly, let $r_i$ be the rightmost/biggest $x$ value of all clients that choose to utilize stack $i$.
\end{definition}

\begin{lemma}
\label{cause bubble utility}
If there is an up (down) bubble on the right (left) hand side of a core interval $I_i$, then $r_i = R(I_i)$ ($l_i = L(I_i)$).
\end{lemma}
\begin{proof}
Let $B_i$ be an up-bubble on the right of a core interval $I_i$. Let $j$ be a stack such that $c(R(I_i))=c(R(I_i),x_j)$, $i\neq j$ and for some small $\epsilon > 0$ we have that $c(R(I_i)+\epsilon)=c(R(I_i)+\epsilon, x_j)$. We know that stack $j$ must be to the left of $R(I_i)$ because stack $i$ is followed by an up-bubble and if $j$ were to the right of $R(I_i)$ the cost function would begin to decrease after $R(I_i)$ instead of increase. We also know that stack $j$ must be to the left of $L(I_i)$ or else $I_i$ would not be a core interval. Since $B_i$ is an up-bubble, we know that $c(R(I_i)+\epsilon,x_j)<c(R(I_i)+\epsilon,x_i)$ for some small $\epsilon>0$ because otherwise the position $R(I_i)+\epsilon$ would still be within the core interval of $i$, causing a contradiction. We also know that $n_i<n_j$ because $d(x_j,R(I_i)+\epsilon)>d(x_i,R(I_i)+\epsilon)$ and therefore the only way for $d(x_j,R(I_i)+\epsilon)/n_j<d(x_i,R(I_i)+\epsilon)/n_i$ is if $n_i<n_j$ holds. Thus, for all points $x'\geq R(I_i)+\epsilon$ we know that $c(x',x_j)<c(x',x_i)$ because both are linear functions, but $c(x',x_j)$ has both a lower initial value and a smaller slope (the slopes are defined as $\frac{1}{n_i}$ and $\frac{1}{n_j}$ respectively by the definition of the cost function). Therefore, the cost function of $i$ will never be minimal in this range and $i$ cannot receive utility from any clients, and likewise bubbles, to the right of its core interval.

The second part of the lemma follows from a symmetric argument for down-bubbles.
\end{proof}

\begin{definition}
Let $\lambda_i= c(l_i)$ for any stack $i$. Similarly, let $\mu_i= c(r_i)$ for any stack $i$.
\end{definition}

Note that clearly $l_i$ is to the left of (or equal to) $L(I_i)$, and $r_i$ is to the right of (or equal to) $R(I_i)$. This is because all clients in the core interval $I_i = [L(I_i),R(I_i)]$ use stack $i$, but there may also be clients outside of the core interval which use stack $i$, i.e., the clients in bubbles. This also immediately tells us that  $\lambda_i\geq\Lambda_i$ and $\mu_i\geq M_i$, since $\lambda_i$ and $\mu_i$ are farther from stack $i$'s location $x_i$. Recall Definition \ref{def:c-A} of $c_{-A}$. Then, we can show the following lemma.

\begin{lemma}[The $\lambda-\mu$ Lemma]
\label{lm lemma}
If there is an up-bubble immediately to the right hand (positive) side of an interval $I_i$, then $\lambda_i<\mu_i$. For any $x\in [l_i,r_i]$, let $\delta = x-l_i$ and let the set $A$ be defined as the set of all stacks positioned in the range $[l_i,r_i]$. Then $c_{-A}(x) \geq \lambda_i+\delta\left[\frac{\mu_i-\lambda_i}{r_i-l_i}\right]$ (i.e., $c_{-A}(x)$ in the interval $[l_i,r_i]$ is at least the linear interpolant of $\lambda_i$ and $\mu_i$) and $c_{-A}(x)$ is concave.
\end{lemma}

\begin{figure*}[ht]
        \centering
        \includegraphics[width=0.8\textwidth]{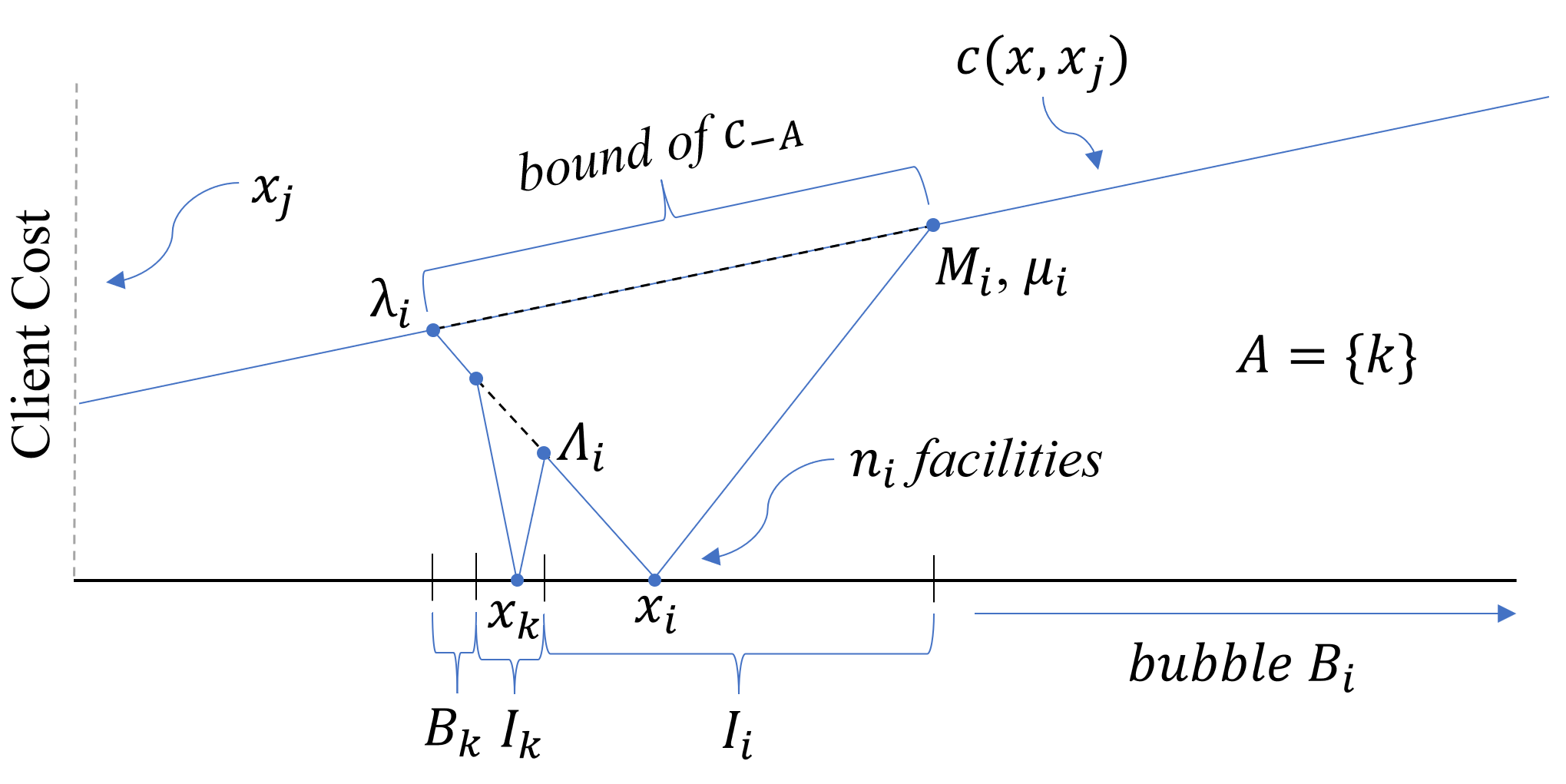}
        \caption{An illustration of the proof of Lemma \ref{lm lemma}.}
        \label{fig:lm}
\end{figure*}

\begin{proof}
See Figure \ref{fig:lm} for an illustration of this proof. Let stack $j$ be the stack such that $\mu_i=c(r_i)=c(r_i, x_j)$, $i\neq j$ and for some small $\epsilon > 0$ we have that $c(r_i+\epsilon)=c(r_i+\epsilon, x_j)$. We know that $R(I_i)=r_i$ due to Lemma \ref{cause bubble utility}. We know that stack $j$ must be to the left of $r_i$ because stack $i$ is followed by an up-bubble and if $j$ were to the right of $r_i$ the cost function would begin to decrease after $r_i$ instead of increase. We also know that stack $j$ must be to the left of $L(I_i)$ or else $I_i$ would not be a core interval. Since $j$ receives utility in the up-bubble, it must have $n_j>n_i$. Let $x_p$ be defined as $c(x_p,x_i)=c(x_p,x_j)$ and $x_p<x_i$. We can bound the utility received by stack $i$ to the left of $x_i$ using stack $j$ because we can show that for all $x'<x_p$ we have that $c(x',x_j)<c(x',x_i)$. From right to left in the interval $[x_j,x_p)$ we know that $c(x',x_j)$ is decreasing while $c(x',x_i)$ is increasing and they start with equal values, so in this range $c(x',x_j)<c(x',x_i)$, and from right to left in the interval $(0,x_j]$ we know that $c(x',x_j)$ increases slower than $c(x',x_i)$ due to the fact that $n_j>n_i$ so the cost $c(x', x_j)$ will never surpass the initial value of stack $i$'s cost function in this reverse interval either; thus, for all $x'<x_p$ we have that $c(x',x_j)<c(x',x_i)$. Thus, we know that $l_i$ is at least $x_p$ which would mean that $\lambda_i\leq c(x_p,x_j)$. This means that $\lambda_i<\mu_i$ because $c(x_p,x_j)<c(r_i,x_j)$ since the cost function of $j$ is increasing in this range.

We can show that any stack in the set $A$ cannot have an individual cost function that affects the global cost function $c(x)$ outside of the range $[l_i,r_i]$. If that stack is $i$, then, by the definition of $l_i$ and $r_i$, it will not have utility outside of that range. If the stack is stack $k\in A$ where $k\neq i$, then it must have $n_k<n_i$ or else stack $i$ would never receive utility beyond point $x_k$ which it must by the definition of $A$. The cost of $c(l_i,x_i)\leq c(l_i,x_k)$ by the definition of $l_i$ and the cost of $c(r_i,x_i)\leq c(r_i,x_k)$ by the definition of $r_i$. Thus, for any $x'<l_i$ or $x'>r_i$ we know that $c(x',x_i)<c(x',x_k)$, since $n_k<n_i$ and both of these cost functions are increasing in these ranges. Since stack $i$ receives no utility in the ranges $x'<l_i$ and $x'>r_i$ and $c(x',x_i)<c(x',x_k)$, we can conclude that stack $k$ also does not receive any utility in these ranges and therefore does not affect the global cost function. Thus, we have shown that this holds for all stacks in $A$.

We know that $c_{-A}(l_i)=\lambda_i$ and that $c_{-A}(r_i)=\mu_i$ by the definition of the cost function and $c_{-A}(x)$, since we have argued that stacks in $A$ cannot have minimum cost outside of interval $[l_i,r_i]$, and thus there must be at least one stack which is outside of $A$ which has minimum cost at $l_i$, and similarly at $r_i$. Now, apply Lemma \ref{monotonic} to the function $c_{-A}(x)$ and interval $[l_i,r_i]$. Since the interval $[l_i,r_i]$ only contains facilities in $A$ by definition of $A$, we know that $c_{-A}(x)$ is concave in this interval. 
Therefore, using the concavity of $c_{-A}(x)$ alongside its values at the start and end of the interval $[l_i,r_i]$, we can conclude that $c_{-A}(x) \geq \lambda_i+\delta\left[\frac{\mu_i-\lambda_i}{r_i-l_i}\right]$ by the definition of concave functions. This can be interpreted as $c_{-A}(x)$ is greater than the linear interpolant of $\lambda_i$ and $\mu_i$ within $[l_i,r_i]$.
\end{proof}

\begin{corollary}
\label{lm corollary}
If there is a down-bubble immediately to the left hand (negative) side of an interval $i$, then $\mu_i<\lambda_i$. For any $x\in [l_i,r_i]$, let $\delta = r_i-x$ and let the set $A$ be defined as the set of all stacks positioned in the range $[l_i,r_i]$. Then $c_{-A}(x) \geq \mu_i-\delta\left[\frac{\mu_i-\lambda_i}{r_i-l_i}\right]$ (i.e., $c_{-A}(x)$ in the interval $[l_i,r_i]$ is at least the linear interpolant of $\lambda_i$ and $\mu_i$) and $c_{-A}(x)$ is concave.
\end{corollary}

\section{Price of Anarchy}\label{sec:poa}
In this section, we present our main result, which is a bound on the Price of Anarchy and thus on the quality of all Nash equilibrium solutions in our model. We begin with the following very simple property.

\begin{proposition}\label{claim:balanced}
In a Nash equilibrium, the utility of any facility is at most a factor of 2 away from the utility of any other. More precisely, given stacks of size $n_i$ and $n_j$, it must be that $$\frac{u_i}{n_i}\geq \frac{u_j}{n_j+1} \geq \frac{1}{2}\cdot  \frac{u_j}{n_j}.$$ 
\end{proposition}
\begin{proof}
This is simply because a facility from the first stack can deviate by moving on top of the second stack $j$ (which would only increase the amount $u_j$), and get utility at least $\frac{u_j}{n_j+1}$. However, a facility should not be able to improve its utility in a Nash equilibrium.
\end{proof}

If there were no bubbles and each stack had to be located in the middle of its core interval, the above proposition could be used to form a nice bound on the Price of Anarchy. If $u_i$ clients use a stack of size $n_i$, we could simply say that the total cost of these clients equals
$$\int_{x=0}^{u_i}\frac{|x-\frac{u_i}{2}|}{n_i}~dx = \frac{u_i^2}{(4n_i)}.$$
In our model, however, just because we know approximately how many clients use a stack does not mean that we have a good bound on the cost that these clients experience. These clients may be far away from the stack due to being in a bubble, or the stack could be off-center in its core interval, which increases the client cost. Because of this, our main goal in this section becomes that of proving a limit on how much bubbles can affect things, and in particular proving a limit on the possible sizes of the bubbles.

\subsection*{Bounding Sizes of Bubbles}

Consider the following assignment of bubbles to adjacent core intervals. Assign each up-bubble to the core interval immediately to the left of the bubble and assign each down-bubble to the core interval immediately to the right of the bubble. We first show that every core interval has at most one bubble assigned to it.

\begin{lemma}
\label{assignment}
The assignment of bubbles to core intervals defined above is an injective mapping, and thus every core interval has at most one bubble assigned to it.
\end{lemma}
\begin{proof}
It is important to note that due to the definition of the cost function, there cannot be up-bubbles at the leftmost point of our interval $X$. This is because in order for there to be an up-bubble at the leftmost point of $X$, there must be at least one stack whose individual cost function is a part of the bubble. Since the cost function in an up-bubble is increasing, this stack must be strictly to the left of the bubble. This would mean that the stack would have to be before the leftmost point of $X$, making this impossible. Likewise, there cannot be any down-bubbles at the rightmost point of $X$ for the same reasons. Thus, this mapping is well defined with every bubble being assigned a core interval.

Now to prove the lemma; this is a proof by contradiction. See Figure \ref{fig:LM} for a helpful illustration. Assume that core interval $i$ has both an up-bubble on the right hand (positive) side and a down-bubble on the left hand (negative) side. Due to Lemma \ref{LM lemma} we can use the up-bubble to infer that $\Lambda_i<M_i$. Due to Corollary \ref{LM corollary} we can use the down-bubble to infer that $M_i<\Lambda_i$. Thus, we have a contradiction. Since there cannot be any core interval with an up-bubble on the right hand side and a down-bubble on the left hand side, no core interval will be assigned two bubbles. Thus, this mapping is injective.
\end{proof}

\begin{definition}
Let $B_i$ be defined as the interval of the bubble assigned to core interval $I_i$ by the above mapping.
\end{definition}

Note that the clients in $B_i$ do {\em not} use the facilities in $I_i$ in this solution: they use some facilities farther away; $I_i$ is simply the core interval next to the bubble $B_i$. The key component of our proof is the fact that the size of $B_i$ cannot be too large compared to $I_i$. Consider Example \ref{ex:bubbles}. The bubble on the right has size $1-\varepsilon$, and would be assigned to the core interval $I_2$, which has size of approximately 4: 2 to the right of $x_2$ and a bit less than 2 to the left of $x_2$. Thus, the size of the bubble is only about $1/4$ of the size of the core interval next to it. In the following, we prove that this is, in fact, the worst case, and no larger bubbles are possible. We do this through a series of lemmas as follows.

\begin{lemma}
\label{cause bubble utility redux}
If a stack $i$ is assigned an up (down) bubble $B_i$, then stack $i$ does not receive utility from any bubbles to its right (left) hand side.
\end{lemma}
\begin{proof}
This is a trivial consequence of Lemma \ref{cause bubble utility}.
\end{proof}

\begin{lemma}
\label{ss bubble}
Singleton stacks (i.e., stacks with $n_i=1$) cannot receive clients or utility from a bubble. In other words, singleton stacks receive all of their utility from clients within their core interval: $U_i=I_i$.
\end{lemma}
\begin{proof}
This is a proof by contradiction. Assume that stack $i$ is a singleton stack that (without loss of generality) receives utility from an up-bubble. Since it is receiving utility within a bubble, then $c(x,x_i)$ within this bubble is the minimum of every stack's cost function at some point $x'$. Since bubbles are disjoint from core intervals, we know that $x'$ is not within $i$'s core interval. Let $x'$ be within bubble $B_j$. We know that $i\neq j$ because if $i=j$, then the client at position $x'$ could not choose stack $i$ due to Lemma \ref{cause bubble utility}. We can conclude that there must be a core interval between $x_i$ and $x'$ because bubble $B_j$ is assigned to core interval $j$ that must exist to the left of $B_j$ and to the right of $I_i$ due to Lemma \ref{assignment}. Since $x'$ is within $B_j$, we can express it as $x'=L(B_j)+\delta=R(I_j)+\delta$. We can therefore express the cost function of $i$ at $x'$ as $c(L(B_j)+\delta,x_i)$ and the cost function of $j$ at $x'$ as $c(L(B_j)+\delta,x_j)$. Since $x_i<x_j$ and $n_i\leq n_j$ since $n_i=1$, we thus know that the cost function of $j$ at $x'$ will always be less than the cost function of $i$ at $x'$. Therefore, we have a contradiction since stack $i$ will not receive utility at position $x'$. Therefore, singleton stacks cannot receive utility from bubbles.
\end{proof}

The above lemma states that singleton stacks will never receive utility from a bubble. The following lemma establishes, further, that singleton stacks will never be {\em assigned} a bubble in a Nash equilibrium, meaning that the core interval of a singleton stack will never have an up-bubble directly to the right of it, or a down-bubble directly to the left of it. Although some singleton stacks may exist in the solution, they will not be assigned any bubbles in our mapping.


\begin{lemma}
If the current solution is a Nash equilibrium, then for each $B_i$ we have that $n_i\geq2$. 
\label{lem:no_single_stack}
\end{lemma}

\begin{figure*}[ht]
        \centering
        \includegraphics[width=0.8\textwidth]{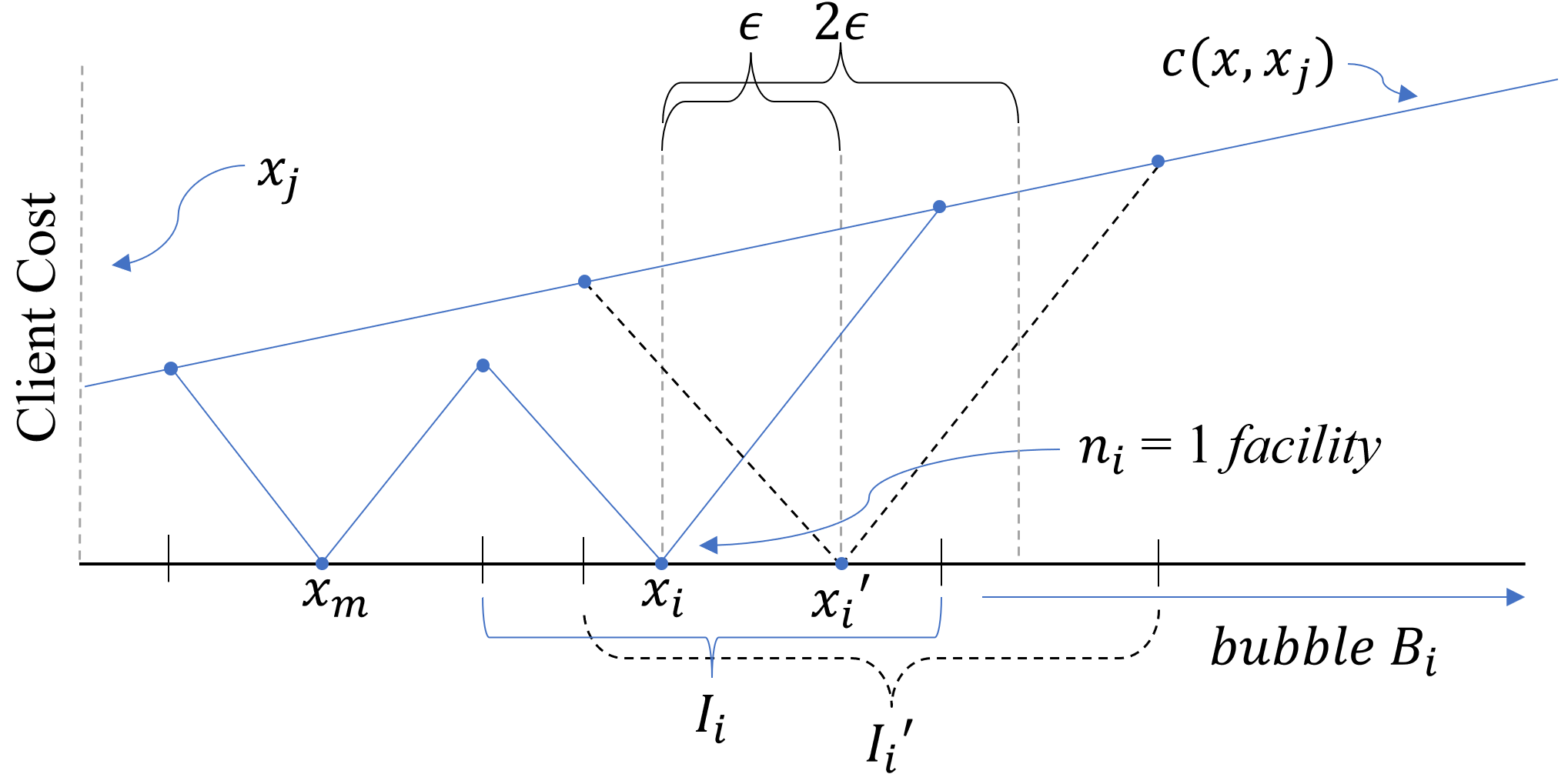}
        \caption{An illustration of the proof of Lemma \ref{lem:no_single_stack}.}
        \label{fig:no_single_stack}
\end{figure*}

\begin{proof}
This is a proof by contradiction. Without loss of generality, assume that there is an up-bubble, $B_i$, which is assigned to a stack, $i$, with $n_i=1$ in a pure Nash equilibrium. We will show that stack $i$ gains more utility by moving from position $x_i$ to position $x_i+\epsilon$ for some small $\epsilon>0$ and $\epsilon< |B_i|/2$. Therefore, this will not be a pure Nash equilibrium and will cause a contradiction. See Figure \ref{fig:no_single_stack} for an illustration of this proof.

Let $j$ be the stack such that $\mu_i=c(r_i)=c(r_i, x_j)$, $j\neq i$ and for some small $\gamma > 0$ we have that $c(r_i+\gamma)=c(r_i+\gamma, x_j)$. Let $m$ be the stack such that $\lambda_i=c(l_i)=c(l_i, x_m)$, $m\neq i$ and for some small $\gamma > 0$ we have that $c(l_i-\gamma)=c(r_i-\gamma, x_m)$. Note that $\Lambda_i=\lambda_i= c(L(I_i))=c(l_i)$ and $M_i=\mu_i= c(R(I_i))=c(r_i)$ due to Lemma \ref{ss bubble}. Let $\epsilon$ be small enough such that for every $\delta\in(0,2\epsilon)$, we have that $c(R(I_i)+\delta) = c(R(I_i)+\delta, x_j)$ and $c_{-i}(L(I_i)+\delta) = c_{-i}(L(I_i)+\delta, x_m)$. Let $i'$ represent stack $i$ after moving. Let $x_i'$ represent the position of the singleton stack $i'$ ($x_i' = x_i + \epsilon$) and let $I_i'$ represent the core interval of stack $i'$. We know that stack $i$ receives all of its utility from its core interval due to Lemma \ref{ss bubble}, so it is sufficient to only consider $I_i$ and $I_i'$.

Let $y$ be the point where $c(y,x_i')=c(y,x_j)$. First, we will show that all clients in the range $[R(I_i),\min\{y,R(I_i)+2\epsilon\}]$ will choose stack $i'$. By our assumption on the size of $\epsilon$ we know that $c_{-i}(x)=c_{-i}(x,x_j)$ for $x\in [R(I_i),R(I_i)+2\epsilon]$. Thus, for $x\in [R(I_i),R(I_i)+2\epsilon]$, we know that each client will use either stack $j$ or stack $i'$, since the costs of using all other stacks are larger than the cost of using stack $j$. Since our definition of $y$ is the point at which stack $i'$ will have a higher cost function than stack $j$, then we know that stack $i'$ will receive all of the clients in this interval range starting at $R(I_i)$ and ending at $y$. Therefore, we know that all clients in the range $[R(I_i),\min\{y,R(I_i)+2\epsilon\}]$ are assigned to $i'$.

We can also show that $y>R(I_i)+\epsilon$. This is because at position $R(I_i)$ we know that $c(R(I_i))=c(R(I_i),x_i')$ because by definition $c_{-i}(R(I_i))=\mu_i$ and $c(R(I_i),x_i')=\mu_i-\epsilon$ because $i$ is a singleton stack. Since the cost function of $j$ is increasing in $B_i$, because $B_i$ is an up-bubble, stack $i'$ will have the lowest cost function for an interval of size $\epsilon$ until it will have cost $\mu_i$ and then will continue to be the lowest cost function for some nonzero interval because the cost function of $j$ has increased as well. Thus, we have that $y>R(I_i)+\epsilon$.

Now we will show that all clients in the range $[L(I_i)+\epsilon,R(I_i)]$ will choose stack $i'$. Note that we can assume that $c(x,x_i)=d(x,x_i)$ and $c(x,x_i')=d(x,x_i')$ because $n_i=1$. By Lemma \ref{LM lemma}, in this region for a point $p=L(I_i)+\delta$, we know that $c_{-i}(p)\geq \Lambda_i+\delta\left[\frac{M_i-\Lambda_i}{|I_i|}\right]$ . Within the sub-range $p\in [L(I_i)+\epsilon,x_i']$ the cost function for stack $i'$ is $c(p,x_i')=d(p,x_i')=x_i'-(L(I_i)+\delta)=[x_i-L(I_i)]+\epsilon-\delta=
\Lambda_i+\epsilon-\delta\le\Lambda_i$, which is less than the bound above, which means that all the clients in this sub-range choose stack $i'$. Next, within the sub-range $p\in [x_i', R(I_i)]$ the cost function for stack $i'$ is $c(p,x_i')=d(p,x_i')=(L(I_i)+\delta)-(L(I_i)+d(L(I_i),x_i'))=(L(I_i)+\delta)-(L(I_i)+\Lambda_i+\epsilon)=\delta-\epsilon-\Lambda_i$. This quantity is less than the bound on $c_{-i}(p)$ above because $|I_i|=M_i + \Lambda_i$ due to $i$ being a singleton stack and $\delta\leq |I_i|$ by definition, which implies,
$$\delta-\epsilon-\Lambda_i = \left[\frac{M_i+\Lambda_i}{|I_i|}\right]\delta-\epsilon-\Lambda_i = \left[\frac{M_i-\Lambda_i}{|I_i|}\right]\delta + \left[\frac{2\Lambda_i}{|I_i|}\right]\delta - \epsilon-\Lambda_i\le \left[\frac{M_i-\Lambda_i}{|I_i|}\right]\delta + 2\Lambda_i - \epsilon-\Lambda_i.$$
This means that the clients in this sub-range all choose stack $i'$ as well. Thus, we have shown that all clients in the range $[L(I_i)+\epsilon,R(I_i)]$ will choose stack $i'$.

Since we have shown that all clients in the range $[L(I_i)+\epsilon,\min\{y,R(I_i)+2\epsilon\}]$ would choose stack $i'$ and that $y>R(I_i)+\epsilon$, we can conclude that $|I_i'|>|I_i|$ since $|I_i|=|[L(I_i)+\epsilon,R(I_i)+\epsilon]|$. Since $|I_i'|$ is strictly larger than $|I_i|$, the utility received by stack $i$ increased, giving it an incentive to deviate. Therefore, this is not a pure Nash equilibrium, causing a contradiction. Since $n_i\neq 1$ for any given $B_i$, we now know that for each $B_i$ we have $n_i\geq 2$.
\end{proof}

Using the above lemmas, we are able to prove the key property which makes our bound on the Price of Anarchy possible:

\begin{lemma}
\label{1 4}
If the current solution is a Nash equilibrium, then for each bubble $B_i$ we have that $|B_i|\leq\frac{1}{4}u_i$.
\end{lemma}
\begin{proof}
Without loss of generality, let $B_i$ be an up-bubble. Let $Z_i$ be defined as the utility received by stack $i$ within the interval $[l_i,x_i]$. Let $z_i=\frac{Z_i}{n_i}$ (i.e., the amount of utility received by each facility in stack $i$ in that range).

 Consider what happens if a single facility $j$ dissents from stack $i$ and moves to position $x_j=R(I_i)$. Note that $r_i=R(I_i)$ for stack $i$, due to Lemma \ref{cause bubble utility}. We will use variables $n_i$, $u_i$, $\lambda_i$, $\mu_i$, $z_i$, $l_i$, $r_i$ and $I_i$ to represent the values of these variables before facility $j$ moves. By the definitions of $z_i$ and $\mu_i$ we know that $\mu_i=c(R(I_i),x_i) = d(R(I_i),x_i)/n_i$ and thus $d(R(I_i),x_i)=n_i\mu_i$. Therefore, before dissenting, facility $j$ received $z_i+\mu_i$ utility because stack $i$ does not receive any utility to the right of $R(I_i)$ due to Lemma \ref{cause bubble utility redux}. This also means that the new location $x_j$ of stack $j$ is $x_j=R(I_i)=x_i+d(x_i,R(I_i)) = x_i+n_i\mu_i$. 

\begin{figure*}[ht]
        \centering
        \includegraphics[width=0.8\textwidth]{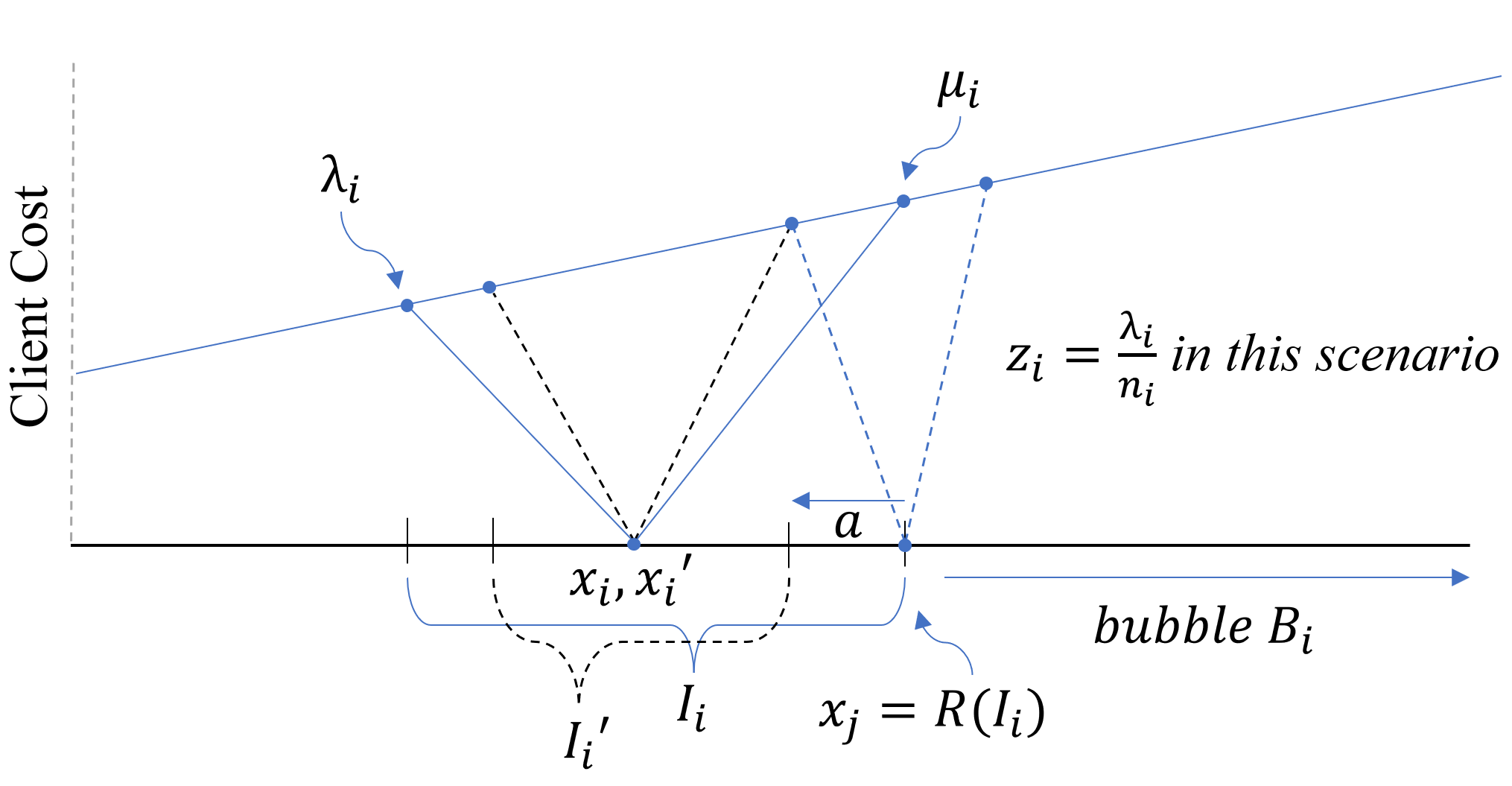}
        \caption{An illustration of the proof of Lemma \ref{1 4}.}
        \label{fig:1 4}
\end{figure*}

Let $i'$ represent the stack $i$ after facility $j$ dissents; it is still at location $x_{i'}=x_i$ but has size $n_{i'}=n_i-1$. Let $a$ be defined as the utility stack $j$ will receive on its left hand side (i.e., the size of the interval $[l_j,x_j]$). See Figure \ref{fig:1 4} for an illustration of the notation in this proof.

$$\text{Let }\alpha_i=\left[\frac{n_i(\lambda_i+\mu_i)}{n_i(\lambda_i+\mu_i)+\mu_i-\lambda_i}\right]$$
We will show that $a\geq\alpha_i\mu_i$ by considering the interval $[x_j-\alpha_i\mu_i,x_j]$ (denoted by $V$), and showing that stack $j$ has the smallest cost function within this range. Since $\lambda_i<\mu_i$ due to Lemma \ref{lm lemma}, we know that $\alpha_i<1$. Thus, within interval $V$, the cost function of $j$ is strictly less than $\mu_i$. The cost function of $i'$ in this range is at all points greater than $\mu_i$ because at position $x_j-\mu_i$ the cost function is $c(x_j-\mu_i,x_{i'})=
\frac{d(x_j-\mu_i,x_i)}{n_i-1} = 
\frac{x_i+ n_i\mu_i-\mu_i - x_i}{n_i-1}=\mu_i$ and the cost function within interval $V$ is strictly larger since it is farther from $x_i$. The cost  functions of all other stacks in $A$ (i.e., stacks located between $l_i$ and $r_i$) also have cost functions larger than $\mu_i$ in this interval $V$. To see why this is true, first note that all such stacks $x_m$ must have values of $n_m$ such that $n_m<n_i$ (or $n_m\leq n_{i'}$), since otherwise the stack $x_i$ would never have obtained any utility beyond point $x_m$, contradicting the definition of $l_i$ and $r_i$. Since $r_i=R(I_i)$, all such stacks $x_m$ are positioned to the left of $x_i$ and they all have values of $n_m\leq n_{i'}$. Thus, their cost functions will always be larger than the cost function of $i'$ in interval $V$, and thus will all be larger than the cost function of $j$. Lastly, we can also show that $c_{-A}(x)\geq c(x,x_j)$ within $V$ due to Lemma \ref{lm lemma}. Let $x=x_j-\alpha_i\mu_i$.
\begin{align*}
    c_{-A}(x) &\geq \mu_i-\alpha_i\mu_i\left[\frac{\mu_i-\lambda_i}{r_i-l_i}\right] \\
    &=\mu_i-\alpha_i\mu_i\left[\frac{\mu_i-\lambda_i}{n_i(\lambda_i+\mu_i)}\right] \\
    &=\mu_i\left(1-\left[\frac{n_i(\lambda_i+\mu_i)}{n_i(\lambda_i+\mu_i)+\mu_i-\lambda_i}\right]\left[\frac{\mu_i-\lambda_i}{n_i(\lambda_i+\mu_i)}\right]\right) \\
    &=\mu_i\left(1-\left[\frac{\mu_i-\lambda_i}{n_i(\lambda_i+\mu_i)+\mu_i-\lambda_i}\right]\right) \\
    &=\mu_i\left[\frac{n_i(\lambda_i+\mu_i)}{n_i(\lambda_i+\mu_i)+\mu_i-\lambda_i}\right] \\
    &=\alpha_i\mu_i
\end{align*}
Since within interval $V$, we now know that $c_{-A}(x)\geq\alpha_i\mu_i$ and $c(x,x_j)\leq\alpha_i\mu_i$, then we know that $c(x,x_j) \leq c_{-A}(x)$ for all $x\in V$. Thus, the cost function of $x_j$ is the minimum one on interval $V$, as compared to the cost functions of $i'$, all stacks in $A\setminus \{i\}$, and $c_{-A}(x)$, so we can conclude that this entire interval uses stack $j$ and thus $a\geq\alpha_i\mu_i$.

Since this solution is a Nash equilibrium, we know that facility $j$ must not receive more utility after dissenting than it did beforehand. Consider the following two exhaustive cases:
\begin{enumerate}[I.]
    \item $|B_i|\leq\mu_i$
    \item $|B_i|>\mu_i$
\end{enumerate}
For Case I, the utility received by facility $j$ will be at least $|B_i|+\alpha_i\mu_i$. This is because all of the clients in $B_i$ will be assigned to facility $j$ because facility $j$ is positioned at $L(B_i)$, $c(L(B_i))=\mu_i$, and the cost function in $B_i$ is increasing. Therefore, after moving to $L(B_i)$, facility $j$ will have the smallest cost function within at least all of $B_i$ since it will begin at $0$ and grow to less than $\mu_i$ through $B_i$. Thus, $|B_i|+\alpha_i\mu_i\leq z_i+\mu_i$ due to the Nash equilibrium qualification, which means that we can bound $|B_i|$ with the following inequality:
$$|B_i|\leq \mu_i\left[1-\alpha_i\right]+z_i.$$
For Case II, we will show that the utility received by facility $j$ will exceed $z_i+\mu_i$, which means that Case II cannot occur in a pure Nash equilibrium. We know that the cost function within $B_i$ is strictly larger than $\mu_i$ because this is an up-bubble where the cost function begins as $\mu_i$. Therefore, we know that the distance between $j$ and the new $R(I_j)$ must be strictly greater than $\mu_i$. Hence, we will let $\mu_i+\gamma$ be defined as the utility received by facility $j$ on its right hand side where $\gamma>0$. Thus, the total utility received by facility $j$ is at least $\alpha_i\mu_i+\mu_i+\gamma$. Since $\gamma$ can be arbitrarily small, to show that this cannot occur in a pure Nash equilibrium, we must show that $z_i\leq\alpha_i\mu_i$. First, we can show that $z_i<\lambda_i$. We have that $d(l_i,x_i)=n_i\lambda_i$ by the definition of $\lambda_i$. Thus, there can be at most $d(l_i,x_i)$ utility for the stack to the left of $x_i$ by the definition of $u_i$. Thus, we have $z_i=\frac{Z_i}{n_i}\leq\frac{n_i\lambda_i}{n_i}=\lambda_i$. Thus, it is sufficient to show that $\lambda_i\leq\alpha_i\mu_i$. We prove this below:
\begin{align*}
\lambda_i &\leq n_i\lambda_i\\
\lambda_i &\leq n_i(\mu_i + \lambda_i)\\
\lambda_i(\mu_i-\lambda_i) &\leq n_i(\mu_i + \lambda_i)(\mu_i - \lambda_i) \\
\lambda_i\mu_i-\lambda_i^2 &\leq n_i\mu_i^2 - n_i\lambda_i^2 \\
n_i\lambda_i^2+\lambda_i\mu_i-\lambda_i^2 &\leq n_i\mu_i^2 \\
n_i\lambda_i^2+n_i\lambda_i\mu_i+\lambda_i\mu_i-\lambda_i^2 &\leq n_i\lambda_i\mu_i+n_i\mu_i^2 \\
\lambda_i(n_i(\lambda_i+\mu_i)+\mu_i-\lambda_i) &\leq \mu_i(n_i(\lambda_i+\mu_i)) \\
\lambda_i &\leq \mu_i\left[\frac{n_i(\lambda_i+\mu_i)}{n_i(\lambda_i+\mu_i)+\mu_i-\lambda_i}\right] \\
\lambda_i &\leq\alpha_i\mu_i
\end{align*}
Therefore, we have shown that $\lambda_i\leq\alpha_i\mu_i$, and by the argument above, Case II cannot occur in a pure Nash equilibrium.

Considering both cases, we now have a bound on $|B_i|$ and therefore a bound on $\frac{|B_i|}{u_i}$ of:
$$\frac{|B_i|}{u_i} = \frac{|B_i|}{n_i z_i+n_i\mu_i}\leq\frac{\mu_i\left[\frac{\mu_i-\lambda_i}{n_i(\lambda_i+\mu_i)+\mu_i-\lambda_i}\right]+z_i}{n_i z_i+n_i\mu_i}.$$
In order to find a bound that will always hold, we must make it as loose as possible. This bound is the loosest when $n_i=2$ because $n_i\geq2$ (due to Lemma \ref{lem:no_single_stack}) and $n_i$ is only in denominators. Thus, we will set $n_i=2$ giving us:
$$\frac{|B_i|}{n_i z_i+n_i\mu_i}\leq\frac{1}{2}\cdot\frac{\mu_i\left[\frac{\mu_i-\lambda_i}{3\mu_i+\lambda_i}\right]+z_i}{z_i+\mu_i}.$$
The derivative of this bound with respect to $\lambda_i$ is always negative, so to achieve the loosest bound, we set $\lambda_i=z_i$ since $z_i$ is a lower bound on $\lambda_i$ giving us:
$$\frac{|B_i|}{n_i z_i+n_i\mu_i}\leq\frac{1}{2}\cdot\frac{\mu_i\left[\frac{\mu_i-z_i}{3\mu_i+z_i}\right]+z_i}{z_i+\mu_i}.$$
Next, the derivative of this bound with respect to $z_i$ is always positive, so to achieve the loosest bound we set $z_i=\mu_i$ since $\mu_i$ is an upper bound on $z_i$ due to Lemma \ref{lm lemma} and the definition of $z_i$ giving us:
$$\frac{|B_i|}{n_i z_i+n_i\mu_i}\leq\frac{1}{2}\cdot\frac{\mu_i\left[\frac{\mu_i-\mu_i}{3\mu_i+\mu_i}\right]+\mu_i}{\mu_i+\mu_i}=\frac{1}{2}\cdot\frac{0+\mu_i}{2\mu_i}=\frac{1}{2}\cdot\frac{1}{2}=\frac{1}{4}.$$
Therefore, since our choice of $i$ was arbitrary, $\frac{|B_i|}{n_i z_i+n_i\mu_i}\leq\frac{1}{4}$. Since $u_i = n_i z_i+n_i\mu_i$, we have shown that $|B_i|\leq\frac{1}{4}u_i$.
\end{proof} \bigskip

\subsection*{Bounding the Price of Anarchy}

Now that we have established that bubbles cannot be too large, we can proceed to bound the costs of the clients in a Nash equilibrium solution. Consider an arbitrary Nash equilibrium solution called $S$, with stack locations $x_i$ and stack sizes $n_i$. Let $c_{S}$ represent the total cost of solution $S$, and let $N$ be the set of all stacks for the solution $S$. Also, let $c_{O}$ be the total cost of the optimal solution; we know by Proposition \ref{claim:proportional} that $c_{O}=\frac{1}{4k}$. To compare $c_{S}$ and $c_{O}$, we compare them to several different intermediate quantities, as follows.

\begin{definition}
Let $c_{1}$ be defined as:
$$c_1=\sum_{i\in N} \frac{u_i^2}{2n_i}$$
where each $u_i$ and $n_i$ represent the same utilities and stack sizes as in solution $S$.
\end{definition}

\begin{definition}
Let $c_{2}$ be defined as:
$$c_2=\sum_{i\in N} \left[\int_{x\in I_i}c(x,x_i)~dx + \int_{x\in B_i}c(x,x_i)~dx\right].$$
\end{definition}

Note that $c_1$ and $c_2$ do not represent values of any actual solutions; they are only convenient values for comparing the total cost of equilibrium and optimal solutions. For intuition on the quantity $c_1$, consider a stack with utility $u_i$, where all the clients using this stack lie in the core interval (i.e., there are no bubbles). Then this interval $I$ has size $|I|=u_i$, and the quantity $\frac{u_i^2}{2n_i}$ is an upper bound on the total cost experienced by the clients of $I$: in the worst case the stack would be located at the very start of interval $I$, and thus the  cost of the clients in this interval is $\int_{x=0}^{u_i}\frac{x}{n_i}dx = \frac{u_i^2}{2n_i}.$ So, in some sense, $c_1$ represents the worst case cost if the $u_i$ values were the same as in $S$, but there were no bubbles.

\begin{lemma}
\label{ordering}
$c_{S}\leq c_{2}$
\end{lemma}
\begin{proof}
The cost of all clients in core intervals remains the same between both total costs, thus, the only difference is in the cost of the clients in bubbles. For any bubble $B_i$, the cost $c(x)$ of any given client $x\in B_i$ must be less than $c(x,x_i)$ or else that client would have chosen stack $i$ which we know is not the case due to Lemma \ref{cause bubble utility}. Thus, the value of $c_2$ must be larger since it is increasing the cost of all clients in bubbles.
\end{proof}

\begin{lemma}
\label{25 16}
$c_{2}\leq\frac{25}{16}c_{1}$
\end{lemma}
\begin{proof}
We will compare the values within both summations for any given $i$. Let $c_{2,i}$ and $c_{1,i}$ be defined as:
$$c_{2,i}=\int_{x\in I_i}c(x,x_i)~dx + \int_{x\in B_i}c(x,x_i)~dx \qquad \text{and} \qquad c_{1,i}=\frac{u_i^2}{2n_i}.$$
We will, therefore, show that $c_{2,i}\leq\frac{25}{16}c_{1,i}$. Without loss of generality, assume that $B_i$ is an up-bubble. Recall that $L(I_i)$ and $R(I_i)$ denote the leftmost and rightmost points of interval $I_i$.
\begin{align*}
    c_{2,i} &= \int_{x\in I_i}\frac{|x-x_i|}{n_i}~dx + \int_{x\in B_i}\frac{|x-x_i|}{n_i}~dx \\
    &= \frac{(x_i-L(I_i))^2}{2n_i} + \frac{(R(I_i)-x_i+|B_i|)^2}{2n_i} \\
    &\leq \frac{(x_i-L(I_i) + R(I_i)-x_i+|B_i|)^2}{2n_i} \\
    &= \frac{(R(I_i)-L(I_i)+|B_i|)^2}{2n_i} \\
    &= \frac{(|I_i|+|B_i|)^2}{2n_i} \tag{by the definition of $I_i$} \\
    &\leq \frac{(u_i+|B_i|)^2}{2n_i} \tag{by the definition of $u_i$} \\
    &\leq \frac{(u_i+\frac{u_i}{4})^2}{2n_i}  \tag{by Lemma \ref{1 4}} \\
    &= \frac{25}{16}\cdot\frac{u_i^2}{2n_i} \\
    &= \frac{25}{16}c_{1,i} \tag{by the definition of $c_{1,i}$}
\end{align*}
Since this is true for all $i\in N$, we have shown that $c_{2}\leq\frac{25}{16}c_{1}$.
\end{proof}

\begin{lemma}
\label{9 4}
$c_{1}\leq\frac{9}{4}c_{O}$
\end{lemma}
\begin{proof}
Using Proposition \ref{claim:balanced}, we know that $u_i/n_i$ cannot be too different from each other, while the optimal solution is the total client cost when $u_i/n_i$ are perfectly balanced, i.e., each equals $1/k$. This proof does not use game theory or anything beyond algebra and calculus and simply bounds the largest value that $c_1$ could have if each $u_i/n_i$ are not too different from each other. We include the entire argument for completeness.

From Proposition \ref{claim:balanced}, we know that for any given stacks $j$ and $m$, we have that both $\frac{u_j}{n_j}\leq2\frac{u_m}{n_m}$ and that $\frac{u_m}{n_m}\leq2\frac{u_j}{n_j}$. We can apply this to the utilities in $c_1$ because they are all from solution $S$ which is a pure Nash equilibrium solution. To bound the value of $c_1$ we consider the worst case distribution of $u_i$ values. This will give a bound on the actual value of $c_1$ because the worst possible distribution of $u_i$ values (\textit{which will not necessarily represent any realizable solution}) will be an upper bound on all $c_1$ values generated from pure Nash equilibria.

We have the following set of constraints on our facility utilities and stack sizes:
$$u_i\in (0,1)\quad n_i\in\mathbb{N}$$
$$\sum_i u_i=1\quad \sum_i n_i=k$$
$$\forall (j,m)\in N\left[\frac{u_j}{n_j}\leq2\frac{u_m}{n_m}
\right]$$
Our goal is to bound the worst-case $c_1=\sum_{i\in N} \frac{u_i^2}{2n_i}$ with the above constraints. 

Consider any values of $u_i$ and $n_i$ which obey the above constraints. Now consider a different set of values where each stack of size $n_i$ with utility $u_i$ is replaced by $n_i$ stacks of size 1, each with utility $u_i/n_i$. These new values still obey all the above constraints, and note that the cost $c_1$ is still the same as before, because the $n_i$ individual costs of $\frac{(u_i/n_i)^2}{2}$ are equivalent to the original stack's cost of $\frac{u_i^2}{2n_i}$. Thus, we can assume that for all $i$, we have $n_i=1$. 

Since we can assume all $n_i=1$, we have now reduced our goal to bounding the maximum $c_1=\sum_{i=1}^k \frac{u_i^2}{2}$ under the constraints 
$$u_i\in (0,1)\quad \sum_{i=1}^k u_i=1$$
$$\forall (j,m)\in N\left[{u_j}\leq2{u_m}
\right]$$

Let $v_{\min}=\min_i u_i$. Then, to maximize $c_1$ under the above constraints, it must be that for at least $k-1$ values $u_i$, we have that $u_i\in \{v_{\min}, 2v_{\min}\}$, i.e., all but one $u_i$ value must equal either $v_{\min}$ or $2v_{\min}$. To prove this, suppose to the contrary that, when $c_1$ is maximum, we have some $u_i$ and $u_j$ with $j\neq i$ and $u_i\geq u_j$ are contained in $(v_{\min},2v_{\min})$. Then, consider how the cost $c_1$ changes if we add a small value $\epsilon$ to $u_i$ and subtract it from $u_j$. We have that
\begin{align*}
  (u_i+\epsilon)^2+(u_j-\epsilon)^2&=u_i^2+2u_i\epsilon+\epsilon^2+u_j^2-2u_j\epsilon+\epsilon^2 \\
  &=u_i^2+u_j^2+(2u_i\epsilon-2u_j\epsilon+2\epsilon^2) \\
  &\geq u_i^2+u_j^2+2\epsilon^2 \\
  &>u_i^2+u_j^2
\end{align*}
Thus, the total cost $c_1$ strictly increases, and we still obey the constraints that all $u_i$ values are within a factor of 2 of each other if we choose a small enough $\epsilon$. This gives a contradiction with our assumption that $c_1$ is maximized. Therefore, we know that at most one $u_i$ value can be strictly between $v_{\min}$ and $2v_{\min}$ when $c_1$ is maximized.


Let $x$ be the number of $u_i$ values which equal $2v_{\min}$, and $y$ be the number of $u_i$ values which equal $v_{\min}$. By the above argument, we know that either $x+y=k$, or $x+y=k-1$ when $c_1$ is maximized.

Let us first consider the case when $x+y=k-1$. We will argue that, in fact, {\em all} $u_i$ values must be equal to either $v_{\min}$ or $2v_{\min}$ when $c_1$ is maximized, so this case is not possible. Suppose to the contrary that $x+y=k-1$ when $c_1$ is maximized. Let $u_{\ell}$ be the single value which is strictly between $v_{\min}$ and $2v_{\min}$. Then, 
$$c_1=\sum_{i=1}^k \frac{u_i^2}{2}=\frac{yv_{min}^2+x(2v_{min})^2+u_{\ell}^2}{2}=\frac{1}{2}(4x+y)v_{min}^2+\frac{1}{2}u_{\ell}^2.$$
Since $\sum_{i=1}^k u_i =1$, we can express $u_{\ell}$ as $1-yv_{\min}-x(2v_{\min})=1-(2x+y)v_{\min}$. Consider decreasing $u_{\ell}$ by some small value $\epsilon$ and increasing $v_{\min}$ by $\frac{\epsilon}{2x+y}$. These new values still satisfy all the constraints: the sum of $u_i$ values equals 1, and all pairs of $u_i$ values are within a factor of 2 of each other as long as $\epsilon$ is small enough. Consider how the cost $c_1$ has changed, however: it is now equal to  
\begin{align*}
c_1 &=\frac{1}{2}(4x+y)\left(v_{min}+\frac{\epsilon}{2x+y}\right)^2+\frac{1}{2}(u_{\ell}-\epsilon)^2\\
&\geq\frac{1}{2}(4x+y)v_{min}^2 +\frac{1}{2}u_{\ell}^2 + \frac{1}{2}\frac{2\epsilon(4x+y)}{2x+y}v_{\min}-\frac{1}{2}2\epsilon u_{\ell} +\epsilon^2.
\end{align*} 
The above cost has increased when $\frac{4x+y}{2x+y}v_{\min} > u_{\ell}$. It also increased when $\frac{4x+y}{2x+y}v_{\min} = u_{\ell},$ due to the $\epsilon^2$ term. We can similarly consider what happens when we increase the value of $u_{\ell}$ by $\epsilon>0$ and decrease $v_{\min}$ by $\frac{\epsilon}{2x+y}$. The new cost $c_1$ now equals:
\begin{align*}
c_1 &=\frac{1}{2}(4x+y)\left(v_{min}-\frac{\epsilon}{2x+y}\right)^2+\frac{1}{2}(u_{\ell}+\epsilon)^2\\
&\geq\frac{1}{2}(4x+y)v_{min}^2 +\frac{1}{2}u_{\ell}^2 - \frac{1}{2}\frac{2\epsilon(4x+y)}{2x+y}v_{\min}+\frac{1}{2}2\epsilon u_{\ell}+\epsilon^2.
\end{align*} 
This increases the cost when $\frac{4x+y}{2x+y}v_{\min} \leq u_{\ell}$. Thus, since for all values of $u_{\ell}$ there exists another valid solution with larger cost $c_1$, then our original solution cannot be maximizing $c_1$. Therefore, we have a contradiction.

We now consider the only remaining case, i.e., when $x+y=k$, and so every $u_i$ either equals $v_{\min}$ or $2v_{\min}$. For this case, $v_{\min} = \frac{1}{2x+y}$ since $\sum_{i=1}^{k} u_i = 1$. When $c_1$ is maximized, we have that 
$$c_1=\sum_{i=1}^k \frac{u_i^2}{2}=\frac{yv_{min}^2+x(2v_{min})^2}{2} =\frac{1}{2}(4x+y)v_{min}^2 = \frac{4x+y}{2(2x+y)^2}.$$ Note that since $x+y=k$, we have that $y=k-x$, so this is the same as maximizing 
\begin{equation}\label{eqn.x+y=k}
    c_1 = \frac{k+3x}{2(k+x)^2}.
\end{equation}
Simple calculus shows that this value is maximized on the range of $x\in [0,k-1]$ when $x=\frac{k}{3}$ (note that $x\leq k-1$ by its definition, since there has to be at least one value which equals $v_{\min}$). Thus, we have that 
$$c_1\leq\frac{k+3x}{2}\left(\frac{1}{k+x}\right)^2\leq\frac{2k}{2}\left(\frac{1}{\frac{4}{3}k}\right)^2=\frac{9}{16k}.$$ By Lemma \ref{lem:opt_cost} we know that
$c_{O}=\frac{1}{4k},$ so overall we have that
$$\frac{c_1}{c_{O}}\leq \frac{9/16k}{1/4k}=\frac{9}{4},$$ as desired.
\end{proof}



\begin{theorem}\label{thm:poa}
The Price of Anarchy in our setting is at most $\frac{225}{64}\approx 3.516$.
\end{theorem}
\begin{proof}
\begin{align*}
    c_{S}&\leq c_{2} \tag{by Lemma \ref{ordering}} \\
          &\leq \frac{25}{16}c_{1} \tag{by Lemma \ref{25 16}} \\
          &\leq \frac{25}{16}\cdot\frac{9}{4}c_{O} \tag{by Lemma \ref{9 4}} \\
          &=\frac{225}{64}c_{O}
\end{align*}
Since our choice of $S$ was arbitrary, this bound holds for all pure Nash equilibria.\end{proof}

\section{Conclusion and Future Directions}
In this paper, we introduced and analyzed a new variation of the Hotelling-Downs model, in which clients do not simply use the closest facility, but instead prefer sites with many facilities. Because of this, interesting interactions between facilities begin occurring: the facilities want to be far away from other facilities so that they get more customers to themselves, but they also want to be close to other facilities since this might increase the amount of customers they get (by creating a ``destination to visit''). This can also apply to social choice settings, where political candidates decide whether they want to position themselves as different from all others (to make themselves unique) or join an existing party (to ``ride on their coat-tails''). Despite Nash equilibria having quite different structure in our setting than in the classic version, we showed that a good Nash equilibrium always exists and that the Price of Anarchy is bounded, thus establishing that Nash equilibrium solutions cannot be very bad as compared with the optimal facility placement.

Our work is only a first step in the study of what we call ``The Mall Effect'', however, and many open questions and future directions remain. The most immediate is looking at other cost functions $c(x,x_i)$ which increase with distance and decrease with the number of facilities at $x_i$. We believe that the function $c(x,x_i)=|x-x_i|/n_i$ used in this paper is natural, but many other cost functions make sense as well. For example, what changes if we consider $c(x,x_i)=|x-x_i|/(n_i)^p$ for some power $p$? For $p<1$, the properties of this model become similar to the classic one, and a Nash equilibrium does not always exist. For $p>1$, Nash equilibria always exist, and, in fact, this should help with the Price of Anarchy as well since this gives further incentive for facilities to form large stacks. However, fully analyzing the Price of Anarchy in this model remains future work. A different type of cost function may give incentive for clients to go to facilities where many other facilities are {\em nearby}, instead of at exactly the same location. Alternatively, another similar variation could be derived by non-uniformly distributing the utility from clients to the set of facilities in a given stack.

Another promising direction involves looking at different metric spaces. Although we focused on $X$ being a one-dimensional interval as in most existing work on this subject, all our results also hold if $X$ were a circle instead. It would be interesting to consider more general metric spaces, such as two-dimensional spaces, or facilities which can be located on a graph such as in \cite{Networks} and \cite{Cycle_Graphs}. It would also be interesting to consider settings where the set of clients is not continuous, but instead there is a discrete set of client locations, as well as a discrete set of possible facility locations, as in \cite{Finite_Locations}. Finally, it would be interesting to combine our client behavior with some of the other variations in the existing literature, such as putting limits on how far clients are willing to travel in order to use a facility: if all facilities are too far, they don't use any facility at all \cite{Variations,Limited_Attraction}.

Our work initiates the study of facility synergy, thereby shedding light on potential social processes that can lead to the creation of malls, shopping centers and downtowns while opening up a new avenue for future research.

\subsubsection*{Acknowledgments}
This work was partially supported by National Science Foundation award CCF-2006286. A conference version of this paper previously appeared in the proceedings of the {\em 18th International Symposium on Algorithmic Game Theory (SAGT 2025)}.

\bibliography{references}

\appendix
\section{Lower Bound on the Price of Anarchy}
While our main result is an upper bound on the price of anarchy, we
unfortunately do not have a matching lower bound; the best lower bound we have is very close to 1. Forming a non-trivial lower bound on the price of anarchy for this model seems to require new insights and techniques, and would be a good direction for future work. When forming lower bounds for this setting, note that proportional facility placements always provide the optimum solution, and for non-proportional solutions with bubbles, even showing that a solution is a Nash equilibrium already requires a significant amount of work, let alone forming a Nash equilibrium with high cost. While our structural results and lemmas above provide many necessary conditions for a solution to be a Nash equilibrium, these conditions do not seem to be sufficient, and thus do not lend themselves to forming lower bound examples. 

Despite our lower bound example being quite simple, we include it below for completeness. 

\begin{proposition}[A lower bound on the price of anarchy]
\label{prop:poa-lower-bound-65-64}
For the Hotelling--Downs game with facility synergy, there exist pure Nash equilibria with cost at least $\frac{65}{64}$ times the optimum cost.  
In other words, the price of anarchy is at least
    $\frac{65}{64}$.
\end{proposition}
\begin{proof}
Consider the case of \(k=4\) facilities on the interval \(X=[0,1]\). Place two facilities at
\[
    a=\frac{7}{32}
\qquad
\text{and two facilities at}
\qquad
    b=\frac{25}{32}.
\]
We first show that this placement is a pure Nash equilibrium.

Since the two stacks have the same size, the client boundary between them is the midpoint
\[
    \frac{a+b}{2}
    =
    \frac{1}{2}.
\]
Thus, the stack at \(a\) receives the clients in \(\left[0,\frac{1}{2}\right]\), and the stack at \(b\) receives the clients in \(\left[\frac{1}{2},1\right]\). Each stack, therefore, receives client mass \(\frac{1}{2}\), and since each stack contains two facilities, each facility obtains utility
\(    \frac{1/2}{2}
    =
    \frac{1}{4}.
\)

It remains to show that no facility can obtain utility strictly larger than \(\frac{1}{4}\) by deviating.  By symmetry, it suffices to consider one of the facilities currently located at \(a\). Suppose this facility deviates to some point \(y\in[0,1]\).  After the deviation, there is one facility remaining at \(a\), two facilities at \(b\), and the deviating facility is alone at \(y\), unless \(y=b\), which we handle separately.

First suppose that \(y<a\). For a client located at \(x\) to choose the deviating facility at \(y\) over the remaining singleton facility at \(a\), it is necessary that
\[
    |x-y|\le |x-a|.
\]
Since \(y<a\), this implies
\[
    x\le \frac{a+y}{2}.
\]
Therefore, the deviating facility can receive client mass at most
\[
    \frac{a+y}{2}
    \le a
    =
    \frac{7}{32}
    <
    \frac{1}{4}.
\]
Hence such a deviation is not profitable.

Next suppose that \(y>b\). To upper bound the utility of the deviating facility, it is enough to consider only the constraint that it must beat the size-two stack at \(b\).  That is, any client choosing \(y\) must satisfy
\[
    |x-y|\le \frac{|x-b|}{2},
\]
and therefore must also satisfy 
\[
|x-y|\le |x-b|
\]
For \(y>b\), by a symmetric argument with the case when \(y<a\), we obtain that the deviating facility can receive client mass at most 
\[
    1-\frac{b+y}{2}
    \le 1-b
    =
    \frac{7}{32}
    <
    \frac{1}{4}.
\]
Hence such a deviation is not profitable either.

Now suppose that \(a<y<b\). A client \(x\) will choose the deviating facility only if the deviating facility beats both the singleton facility at \(a\) and the stack of two facilities at \(b\).

The condition that \(y\) beats \(a\) is
\[
    |x-y|\le |x-a|,
\]
which, since \(a<y\), implies
\[
    x\ge \frac{a+y}{2}.
\]
The condition that \(y\) beats the size-two stack at \(b\) is
\[
    |x-y|\le \frac{|x-b|}{2}.
\]
For \(a<y<b\), the set of points $x$ satisfying this inequality is
\[
    \left[2y-b,\; \frac{b+2y}{3}\right].
\]
Consequently, the deviating facility can receive client mass at most
\[
    \frac{b+2y}{3}
    -
    \max\left\{
        \frac{a+y}{2},
        2y-b
    \right\}.
\]
We now maximize this expression over \(y\in(a,b)\). The two lower bounds are equal when
\[
    \frac{a+y}{2}=2y-b,
\]
or equivalently
\[
    y=\frac{a+2b}{3}.
\]
For \(y\le (a+2b)/3\), the maximum lower bound is \((a+y)/2\), and the expression is increasing in \(y\).  For \(y\ge (a+2b)/3\), the maximum lower bound is \(2y-b\), and the expression is decreasing in \(y\). Hence the maximum occurs at
\[
    y=\frac{a+2b}{3}.
\]
At this point, the maximum client mass obtainable by the deviating facility is
\[
    \frac{b+2y}{3}-\frac{a+y}{2}
    =
    \frac{4(b-a)}{9}.
\]
Substituting
\[
    b-a
    =
    \frac{25}{32}-\frac{7}{32}
    =
    \frac{18}{32}
    =
    \frac{9}{16},
\]
we get
\[
    \frac{4(b-a)}{9}
    =
    \frac{4}{9}\cdot \frac{9}{16}
    =
    \frac{1}{4}.
\]
Thus, a deviation to a point \(y\in(a,b)\) can give utility at most \(\frac{1}{4}\), and hence is not strictly profitable.


It remains to consider the case \(y=b\), where the deviating facility joins the stack at \(b\). After this deviation, there is one facility at \(a\) and three facilities at \(b\). The boundary between the clients choosing \(a\) and those choosing \(b\) is the point \(x\) satisfying
\[
    x-a=\frac{b-x}{3}.
\]
Solving for \(x\) gives
\[
    x=\frac{b+3a}{4}
    =
    \frac{25/32+3(7/32)}{4}
    =
    \frac{46/32}{4}
    =
    \frac{23}{64}.
\]
Therefore, the size-three stack at \(b\) receives client mass
\[
    1-\frac{23}{64}
    =
    \frac{41}{64}.
\]
The deviating facility would receive one third of this utility, namely
\[
    \frac{1}{3}\cdot \frac{41}{64}
    =
    \frac{41}{192}
    <
    \frac{1}{4}.
\]
Thus, joining the other stack at $b$ is also not profitable.

We have shown that no facility at \(a\) can improve its utility by deviating. By symmetry, the same holds for facilities at \(b\). Hence the placement is a pure Nash equilibrium.

We now compute its social cost. Since the left stack serves \(\left[0,\frac{1}{2}\right]\) and the right stack serves \(\left[\frac{1}{2},1\right]\), and since the placement is symmetric, the total client cost is
\[
    2\int_0^{1/2}\frac{|x-a|}{2}\,dx
    =
    \int_0^{1/2}|x-a|\,dx.
\]
Using \(a<\frac{1}{2}\), we obtain
\[
    \int_0^{1/2}|x-a|\,dx
    =
    \int_0^a (a-x)\,dx
    +
    \int_a^{1/2} (x-a)\,dx
    =
    \frac{a^2}{2}
    +
    \frac{(1/2-a)^2}{2}.
\]
Substituting \(a=7/32\), we get
\[
    \frac{a^2}{2}
    +
    \frac{(1/2-a)^2}{2}
    =
    \frac{(7/32)^2}{2}
    +
    \frac{(9/32)^2}{2}
    =
    \frac{49+81}{2048}
    =
    \frac{65}{1024}.
\]
Therefore, this Nash equilibrium has social cost
\[
    c=\frac{65}{1024}.
\]

On the other hand, for \(k=4\), Lemma \ref{lem:opt_cost} tells us that the optimal social cost is
\[
    c_{O}=\frac{1}{4k}=\frac{1}{16}=\frac{64}{1024}.
\]
Consequently,
\[
    \frac{c}{c_{O}}
    =
    \frac{65/1024}{1/16}
    =
    \frac{65}{64}.
\]
This proves that the price of anarchy is at least \(65/64\) for this example.
\end{proof}

\end{document}